\documentclass[11pt]{article}
\usepackage{graphicx}

\usepackage[colorlinks=true,citecolor=blue]{hyperref}
\usepackage{natbib}
\usepackage{graphicx}
\usepackage{amsfonts}
\usepackage{amsmath}
\usepackage{amssymb}
\usepackage{url}
\usepackage{fancyhdr}
\usepackage{indentfirst}
\usepackage{enumerate}
\usepackage{titlesec}
\usepackage{amsthm}
\usepackage{dsfont}
\usepackage[misc]{ifsym}
\usepackage{multicol}
\usepackage{blindtext}
\usepackage{geometry}
\usepackage{appendix}
\usepackage{float}
\usepackage{caption}
\usepackage{mathrsfs}
\usepackage{multirow}
\usepackage{arydshln}

\theoremstyle{definition}

\newtheorem{example}{Example}

\newtheorem{definition}{Definition}

\theoremstyle{plain}
\newtheorem{theorem}{Theorem}
\newtheorem{lemma}{Lemma}
\newtheorem{proposition}{Proposition}

\theoremstyle{remark}
\newtheorem{remark}{Remark}

\makeatletter
\def\blfootnote{\xdef\@thefnmark{}\@footnotetext}
\makeatother

\usepackage{cases}

\theoremstyle{definition}

\def\N{\mathbb{N}}

\def\p{\mathbb{P}}
\def\E{\mathbb{E}}

\def\R{\mathbb{R}}

\def\H{\mathcal{H}}

\def\X{\mathcal{X}}

\def\cx{\mathrm{cx}}

\def\X{\mathcal{X}}

\def\d{\mathrm{d}}

\DeclareMathOperator{\im}{Range}

\newcommand{\VaR}{\mathrm{VaR}}

\DeclareMathOperator*{\supp}{supp}

\newcommand{\axiomplain}[2]{%
  \par\vspace{0.5\baselineskip}% vertical space above
  \noindent\hspace{1em}(#1) #2%
  \par\vspace{0.5\baselineskip}% vertical space below
}

\newcommand{\axiomplaina}[2]{%
  \noindent\hspace{1em}(#1) #2%
  \par\vspace{0.5\baselineskip}% vertical space below
}

\newcommand{\axiomplainb}[2]{%
  \par\vspace{0.5\baselineskip}% vertical space above
  \noindent\hspace{1em}(#1) #2%
}

\newcommand{\axiomplainc}[2]{%
  \noindent\hspace{1em}(#1) #2%
}

\usepackage[onehalfspacing]{setspace}

\usepackage{bm}
\usepackage{tikz}
\usetikzlibrary{arrows.meta,calc,backgrounds}

\pgfdeclarelayer{background}
\pgfsetlayers{background,main}

\def\id{\mathds{1}}

\title{Aggregate then evaluate}
\author{
Zachary Van Oosten\thanks{Department of Statistics and Actuarial Science, University of Waterloo,  Canada. \Letter~{\scriptsize\url{zjvanoos@uwaterloo.ca}}} \and Ruodu Wang\thanks{Department of Statistics and Actuarial Science, University of Waterloo,  Canada. \Letter~{\scriptsize\url{wang@uwaterloo.ca}}}}
\date{\today}

\begin{document}
	\maketitle
	\begin{abstract}
We distinguish two frameworks for decisions under ambiguity: evaluate-then-aggregate (ETA) and aggregate-then-evaluate (ATE). Given a statistic that represents the decision maker's pure-risk preferences (such as expected utility) and an ambiguous act, an ETA model first evaluates the act under each plausible probabilistic model using this statistic and then aggregates the resulting evaluations according to ambiguity attitudes. In contrast, an ATE model first aggregates ambiguity by assigning the act a single representative distribution and then evaluates that distribution using the statistic. These frameworks differ in the order in which risk and ambiguity are processed, and they coincide when there is no ambiguity. While most existing ambiguity models fall within the ETA framework, our study focuses on the ATE framework, which is conceptually just as compelling and has been relatively neglected in the literature. We develop a Choquet ATE model, which generalizes the Choquet expected utility model by allowing arbitrary pure-risk preferences. We provide an axiomatization of this model in a Savage setting with an exogenous source of unambiguous events. The Choquet ATE framework allows us to analyze a wide range of ambiguity attitudes and their interplay with risk attitudes.

	\end{abstract}

\noindent\rule{\textwidth}{0.5pt}

\section{Introduction}

The expected utility (EU) model, originally axiomatized by \citet{vNM47}, has long served as the cornerstone of modern decision theory. The EU model is formulated in the pure-risk setting, where acts are represented by distributions over outcomes with objectively specified probabilities. However, the pure-risk setting is often unrealistic in many applications, as decision makers (DMs) rarely have complete knowledge of the true distributions underlying the alternatives they face.

To move beyond the pure-risk setting, \citet{S54} reformulated acts as functions from states of the world to consequences, rather than as distributions over consequences. This framework is now referred to as the Savage setting. Additionally, by imposing reasonably motivated axioms on a DM's preferences over such acts, Savage showed that  (a) a subjective probabilistic model could be elicited for the DM, transforming the problem back into one of the pure-risk setting and (b) within this elicited pure-risk setting, the DM's preferences are consistent with the EU model. This decision model is known as the subjective expected utility (SEU) model.

The SEU model is particularly attractive because it grounds probability theory in behavioral principles. However, \citet{E61} raised a critical challenge: when agents possess asymmetric information about the likelihoods of events, Savage's sure-thing principle can fail to describe plausible behavior. This insight undermined the behavioral appeal of the SEU model and, more generally, probabilistically sophisticated models \citep{MS92}, which assume that uncertainty can be captured by a single subjective probabilistic model and that acts are evaluated solely through their induced distributions over outcomes. Ellsberg's critique was pivotal in launching the formal study of decision-making under ambiguity. Among the most influential decision models under ambiguity are the Choquet expected utility model of \citet{S89}, the maxmin expected utility model of \citet{GS89}, the $\alpha$-maxmin model of \cite{GMM04}, the smooth ambiguity model of \citet{KMM05}, and the variational preferences model of \citet{MMR06}. These decision models under ambiguity share the important feature that, when ambiguity is absent, i.e., when uncertainty can be described by a single probabilistic model, they collapse back to the SEU model.

While reverting to the SEU model in the absence of ambiguity is consistent with the historical development of decision theory, this restriction limits the scope of these decision models. This follows because numerous alternative pure-risk  decision models, which are either empirically   or normatively appealing in certain contexts, 
have been proposed and widely adopted. Examples include the mean-variance model of \citet{M52}, used in finance; the rank-dependent utility model of \citet{Q82}, the prospect theory of \citet{KT79}, and the cumulative prospect theory of \citet{TK92}, used in behavioral decision theory; the expected shortfall model of \citet{AT02}, used in quantitative risk management. The aim of this paper is to analyze general decision models under ambiguity that reduce to probabilistically sophisticated models in the absence of ambiguity, rather than relying solely on the SEU model. To do this, we consider two frameworks that both provide a clear separation between pure risk and ambiguity.
A main difference between these two frameworks is that whether pure risk or ambiguity is processed first. 

The first framework we discuss is called 
\emph{evaluate-then-aggregate (ETA)}.  In these decision models, given an act, the DM first examines the distributions it produces under various probabilistic models. For each distribution, the DM evaluates it using a statistic consistent with their pure-risk preferences, e.g., the expected-utility statistic in the case of the EU model. These model-specific evaluations are then aggregated across all the relevant probabilistic models using some aggregator. This framework provides a clear separation between a \emph{risk component} (the evaluation at each probabilistic model) and an \emph{ambiguity component} (the aggregation across probabilistic models). This framework includes most of the decision models under ambiguity discussed above. 

The second framework is called \emph{aggregate-then-evaluate (ATE)}, which we introduce in this paper. In this framework, given an act, instead of evaluating each model separately and then aggregating, the DM first aggregates their ambiguity regarding the act into a single distribution. Afterwards, the DM evaluates this distribution using a statistic consistent with their pure-risk preferences. This framework provides a clear separation between an ambiguity component (the aggregation into a single distribution) and a risk component (the evaluation of this distribution).  For example, the Choquet expected utility model with respect to a continuous capacity $\nu$ can be viewed as first mapping each act $X$ to a distribution $Q_X$ whose survival function $S_{Q_X}$ is given by \begin{equation}
    \label{eq:mot}
    S_{Q_X}(x)=\nu(X>x),~~x\in \mathbb{R},
\end{equation} and then applying the expected-utility statistic to these distributions. Allowing for general statistics in the second step is what we refer to as the \emph{Choquet ATE} model.

We illustrate the conceptual differences between the two frameworks by an example adapted from \cite{E61}. An urn contains 90 balls; 30 are red, and the remaining 60 consist of blue and yellow balls in unknown proportions. 
The possible distributions of the outcome of randomly picking a ball from the urn are represented by
$\mu_{x}$ for $x\in \{0,1,\dots,60\}$ that represents the number of blue balls in the urn. 
The DM needs to  rank four acts: $R$ pays $\$4$ when a red ball is drawn, $B$ pays $\$6$ when a blue ball is drawn,
$Y$ pays $\$5$ when a yellow ball is drawn,
and  $C$  pays $\$1$ in all states. Suppose that the DM has a statistic $\gamma$ defined on distributions, which computes the (numerical) utility of pure risks (represented by distributions). Note that $R$ and $C$ have no ambiguity, and hence their distributions $Q_R$ and $Q_C$ are known, and the utilities $\gamma(Q_R)$ and $\gamma(Q_C)$ can be directly computed. On the other hand, both $B$ and $Y$ have 61 possible distributions, one under each $\mu_x$, which are denoted by ${B}_\# \mu_x$ and ${Y}_\# \mu_x$. 
Below, we describe how the ETA and ATE approaches assess $B$ and $Y$. 
\begin{enumerate}
    \item  The ETA approach first evaluates the utility of $B$, denoted by $U_x=\gamma(B_\# \mu_x)$ under  $\mu_x$ for every $x\in \{0,1,\dots,60\}$, and then aggregates the $61$ utilities by a mapping $\rho$ representing the DM's ambiguity attitude into a utility $\rho(U_0,U_1,\dots,U_{60})$. 
    For instance, the utilities may be aggregated using the worst-case method, a weighted average, or a smooth ambiguity model. The same approach is taken to evaluate the utility of $Y$. The resulting aggregate utilities will be compared with the utilities $\gamma(Q_R)$ and $\gamma(Q_C)$ to make a comparison.
    \item The ATE approach begins by assigning distributions $Q_B$ and $Q_Y$ to the acts $B$ and $Y$, and then the utilities  $\gamma(Q_B)$ and $\gamma(Q_Y)$ are computed and compared with $\gamma(Q_R)$ and $\gamma(Q_C)$.
    The distributions $Q_B$ and $Q_Y$ can be summarized by the two numbers $q_B=Q_B(\{6\})$
    and $q_Y=Q_Y(\{5\})$. 
   This approach does not assume $Q_{B}$ and $Q_{Y}$ are consistent with a probabilistic model. For example, it may hold that
    $q_B+q_Y<2/3$ (e.g., when the DM is ambiguity averse) or $q_B+q_Y>2/3$ (e.g., when the DM is ambiguity seeking). The distributions $Q_B$ and $Q_Y$ are constructed by aggregating the beliefs and attitudes toward ambiguity.
\end{enumerate}

Both approaches are conceptually feasible and may yield the same preferences over $\{B,Y,R,C\}$ for some $\gamma$ and ambiguity attitudes. Nevertheless, the ETA approach requires the (mental) computation of 61 models for each of $B$ and $Y$, whereas the ATE approach only requires determining the values $q_B$ and $q_Y$. For instance,  suppose  $q_B=0.25$ and $q_Y=0.28$ in the ATE approach. If $\gamma$ is the expected value (i.e., risk neutral) then we have $B\succ Y\succ R\succ C$;  if $\gamma$ is the expected utility with the utility function $x\mapsto x^{1/2}$ on $\R_+$ then we have $C\succ R\succ Y\succ B$.
In this example, the ATE approach seems conceptually simpler and more natural, though we cannot assert with certainty which approach better describes a real individual's thought process.

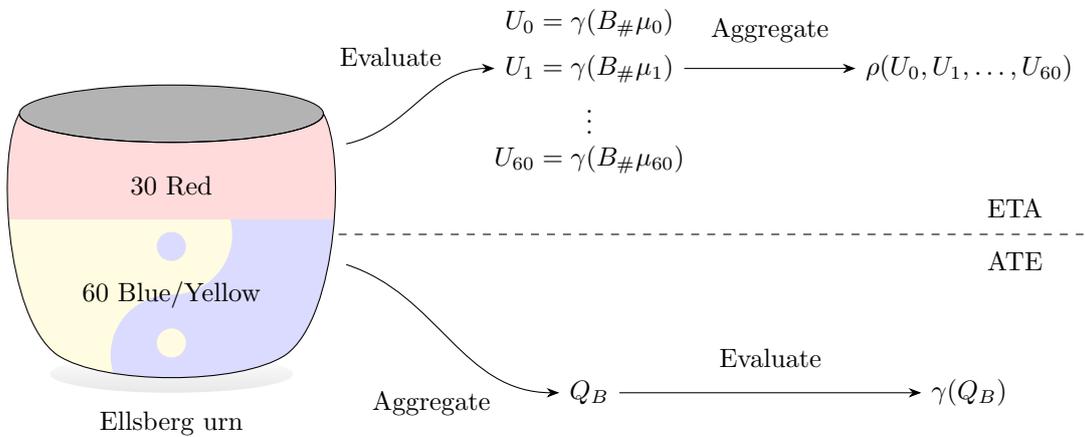
\begin{figure}[ht!]
\centering

% Vertical whitespace ABOVE the urn
\vspace{1em}

\begin{tikzpicture}[>=Stealth, font=\small]

\def\R{2.0}   % radius (x direction)
\def\H{3.2}   % height
\def\E{0.38}  % ellipse y radius
\def\ytop{1.6}
\def\ybot{\ytop-\H}

%-----------------------------------------------
%Urn Body
%-----------------------------------------------

\newcommand{\UrnPath}{%
  (-\R,\ytop) arc (180:360:{\R} and {\E})
    .. controls (\R+0.3,\ytop-0.5) and (\R+0.2,\ybot+0.5) ..
       (0.75*\R,\ybot)
    .. controls (0.45*\R,\ybot-0.4) and (-0.45*\R,\ybot-0.4) ..
       (-0.75*\R,\ybot)
    .. controls (-\R-0.2,\ybot+0.5) and (-\R-0.3,\ytop-0.5) ..
       (-\R,\ytop)
}

 %---------------------------------
  % Bottom region: Yin–Yang style
  %---------------------------------
\begin{scope}
  \clip \UrnPath;
  \begin{scope}
    % restrict to the bottom part of the urn
    \clip (-\R-1,\ybot-2) rectangle (\R+1,0.4);

    % 1) Base: left half yellow, right half blue
    \fill[yellow!14] (-\R-1,\ybot-2) rectangle (0,0.4);
    \fill[blue!14]   (0,\ybot-2)    rectangle (\R+1,0.4);

    % 2) Yin–Yang disk in the middle
    \def\cx{0}        % x-center
    \def\cy{-0.8}     % y-center (tune this)
    \def\ryin{1.6}    % radius   (tune this)

    \begin{scope}
      % limit the following to the Yin–Yang circle
      \clip (\cx,\cy) circle[radius=\ryin];

      % big lobes that carve the S-curve
      \fill[yellow!14] (\cx,\cy+0.5*\ryin)
                       circle[radius=0.5*\ryin];
      \fill[blue!14]   (\cx,\cy-0.5*\ryin)
                       circle[radius=0.5*\ryin];

      % small opposite-colored dots
      \fill[yellow!14] (\cx,\cy-0.4*\ryin)
                       circle[radius=0.12*\ryin];
      \fill[blue!14]   (\cx,\cy+0.4*\ryin)
                       circle[radius=0.12*\ryin];
    \end{scope}
  \end{scope}
%---------------------------------
  % Top region: pastel red as before
  %---------------------------------
  \fill[red!14] (-\R-1,0.2) rectangle (\R+1,\ytop+2);
\end{scope}

% Top of urn
\fill[black!30]
  (-\R,\ytop) arc (180:360:{\R} and {\E})
              arc (0:180:{\R} and {\E}) -- cycle;
\draw
  (-\R,\ytop) arc (180:360:{\R} and {\E})
               arc (0:180:{\R} and {\E});

% Draw the urn
\draw \UrnPath;

\pgfmathsetmacro{\urnbottom}{\ybot - 0.4}

% Shadow BELOW the urn, on the background layer
\begin{scope}[on background layer]
  \shade[
    shading=axis,
    bottom color=black!15,
    top color=transparent!0,
    opacity=0.6
  ] 
    (0,{\urnbottom+0.15}) ellipse ({0.8*\R} and {0.25});
\end{scope}

%-----------------------------------------------
% Text on urn
%-----------------------------------------------
\node at (0,0.65)   {30 Red};
\node at (0,-0.8)  {60 Blue/Yellow};
\node at (0,\ybot-0.9) {Ellsberg urn};

% Anchor points
\coordinate (urnUpper)    at (\R, \ytop-0.4);
\coordinate (urnLower)    at (\R,-0.4);
\coordinate (urnUpperOut) at ($(urnUpper)+(0.3,0)$);
\coordinate (urnLowerOut) at ($(urnLower)+(0.3,0)$);

%-----------------------------------------------
% TOP PATH
%-----------------------------------------------
\node (U0)    at (5.5, 2.8) {$U_0 = \gamma(B_\#\mu_0)$};
\node (U1)    at (5.5, 2.2) {$U_1 = \gamma(B_\#\mu_1)$};
\node (Vdots) at (5.5, 1.6) {$\vdots$};
\node (U60)   at (5.5, 1.0) {$U_{60} = \gamma(B_\#\mu_{60})$};

\node (pU)    at (10.5, 2.2) {$\rho(U_0,U_1,\dots,U_{60})$};

% Arrow urn -> U1
\draw[->] (urnUpperOut) to[out=12,in=180] (U1.west);
\node[above=17pt] at ($(urnUpperOut)!0.3!(U1.west)$) {Evaluate};

% Arrow U1 -> p(U_0,...)
\draw[->] (U1.east) -- (pU.west);
\node[above=6pt] at ($(U1.east)!0.5!(pU.west)$) {Aggregate};

%-----------------------------------------------
% BOTTOM PATH (aligned under U-block)
%-----------------------------------------------
\node (QX)  at (5.5,-2.1) {$Q_B$};
\node (gQX) at (10.5,-2.1) {$\gamma(Q_B)$};

\draw[->] (urnLowerOut) to[out=-18,in=180] (QX.west);
\node[below=25pt] at ($(urnLowerOut)!0.4!(QX.west)$) {Aggregate};

\draw[->] (QX.east) -- (gQX.west);
\node[above=6pt] at ($(QX.east)!0.5!(gQX.west)$) {Evaluate};

%-----------------------------------------------
% ETA / ATE divider
%-----------------------------------------------
\draw[dashed] (\R+0.2,0) -- (12,0);
\node at (11.1, 0.35) {ETA};
\node at (11.1,-0.35) {ATE};

\end{tikzpicture}

\vspace{1em}

\caption{The ETA approach vs the ATE approach for the utility of act $B$}

\end{figure}

The aim of this paper is to make explicit these two different frameworks in the Savage setting, where acts are viewed as mappings from states of the world to real-valued outcomes. Since the ATE framework is new, we will focus on it exclusively in the second half of the paper. The main axiomatic work of this paper will be to characterize the Choquet ATE model. To provide sufficient structure for this axiomatization, we will work in a refined Savage setting with an exogenously given source of unambiguous events. This setting is beneficial because it provides the proper conceptual structure for defining matching probabilities; see \cite{DKW16} and \cite{BHSW18}, which are a major component of the axiomatization, and it allows us to explicitly study the interplay between risk attitudes and ambiguity attitudes. 

We begin in Section \ref{sec:def} by reviewing the mathematical tools that will be used throughout the paper. In Section \ref{sec:sav}, we formally present the Savage setting. In Section \ref{sec:risk}, we formulate the pure-risk setting and provide the mathematical tools for distributions that will support all later discussions of the risk component of preferences. In Section \ref{sec:EA}, we present the ETA model in the Savage setting. In Section \ref{sec:AE}, we do the same with the ATE model. Furthermore,   Theorem \ref{theo:riskEq} provides a novel result characterizing when a mapping assigning acts to distributions is given by \eqref{eq:mot}. In Section \ref{sec:Axioms}, we introduce a refined Savage setting and provide the axioms and  the characterization of the Choquet ATE model in Theorem \ref{theo:mainrep}. Notably, we find that the ambiguity component of the preferences is fully determined by the matching probability. Section \ref{sec:RA} analyzes various ambiguity attitudes for the Choquet ATE model and their interplay with risk attitudes. We start in Section \ref{sec:RAttitudes} by studying various risk attitudes. In Theorem \ref{theo:rD}, we fully characterize the propensity for risk diversification, which will help study its interplay with ambiguity. In Section \ref{sec:AA}, we discuss comparative and absolute ambiguity attitudes and confirm that we get a characterization of these properties through matching probabilities. In Section \ref{sec:div}, we analyze the overall propensity for diversification. We first demonstrate that the propensity for risk diversification alone does not necessarily imply the overall propensity for diversification. Based on  a comonotonic sum inequality, we establish a property for mappings that assign acts to distributions, given in \eqref{eq:coIm}, that guarantees (in the case of the ATE model) the propensity for risk diversification implies the overall propensity for diversification. Finally, Theorem \ref{theo:SAA} establishes that the mapping defined in \eqref{eq:mot} satisfies condition \eqref{eq:coIm} if and only if $\nu$ is supermodular. In Section \ref{sec:robust}, we show how some of the ideas from Section \ref{sec:div} give rise to a distributionally robust representation in Theorem \ref{theo:robust}. We conclude the paper in Section \ref{sec:con}. The proofs of all results, along with related additional results, are provided in the appendices.

\section{Preliminaries}

\label{sec:prelim}

\subsection{Standard definitions and notation}

\label{sec:def}

For a given measurable space $(\Theta,\mathcal{E})$, a function $\nu:\mathcal{E}\mapsto [0,1]$ is a \emph{capacity} if $\nu(\varnothing)=0$, $\nu(\Theta)=1$, and $\nu(A)\leq \nu(B)$ for all $A,B\in \mathcal{E}$ with $A\subseteq B.$ A capacity $\nu$ is said to be: \emph{upward continuous} if for all increasing sequences $(A_n)_{n\in\N}\subseteq\mathcal{E}$ it holds that $\lim_{n\to\infty}\nu(A_n)=\nu\left(\bigcup_{n=1}^\infty A\right)$, \emph{downward continuous} if for all decreasing sequences $(A_n)_{n\in \N}\subseteq\mathcal{E}$ it holds that $\lim_{n\to\infty}\nu(A_n)=\nu\left(\bigcap_{n=1}^\infty A_n\right)$, and \emph{continuous} if it is upward and downward continuous.
A capacity  is called a \emph{probability measure} if for all pairwise disjoint sequences $(A_n)_{n\in \N}\subseteq\mathcal{E}$ it holds that $\nu(\bigcup_{n=1}^\infty A_n)=\sum_{k=1}^\infty \nu(A_n)$. We denote the set of probability measures as $\mathcal{M}_1(\Theta,\mathcal{E})$.

For a continuous capacity $\nu$ , the \emph{core of $\nu$}, denoted by $\mathfrak{C}(\nu)$, is defined as 
$$\mathfrak{C}(\nu)=\left\{\mu\in \mathcal{M}_{1}(\Theta,\mathcal{E}): \mu(A)\geq \nu(A)\text{ for all }A\in \mathcal{E}\right\}.$$   
The following are common properties of a continuous capacity $\nu$. 
\begin{itemize}\setlength{\itemsep}{1pt}
    \item[] Supermodularity: For all $A,B\in \mathcal{E}$, $\nu(A)+\nu(B)\leq \nu(A\cup B)+\nu(A\cap B)$.
    \item[] Submodularity: For all $A,B\in \mathcal{E}$, $\nu(A)+\nu(B)\geq \nu(A\cup B)+\nu(A\cap B)$.
    \item[] Balanced: It holds that $\mathfrak{C}(\nu)\neq \varnothing$.
    \item[] Exactness: For all $A\in \mathcal{E}$, $\nu(A)=\min_{\mu\in \mathfrak{C}(\nu)}\mu(A)$.
\end{itemize}
For more on the core and the previous properties of capacities, see \cite{D94} and \cite{MM04}. 

Denote the set of real-valued bounded Borel functions on $(\Theta,\mathcal{E})$ by $\mathcal{X}(\Theta,\mathcal{E})$. Constant acts in $\mathcal{X}(\Theta,\mathcal{E})$ are identified with constants in $\R$. For $\eta,\zeta\in \mathcal{X}(\Theta,\mathcal{E})$, we write $\eta\geq \zeta$ when $\eta(s)\geq \zeta(s)$ for all $s\in \Theta$.
Given a capacity $\nu$, the \emph{Choquet integral} with respect to $\nu$ is given by
$$\int \eta\d\nu=\int_0^\infty \nu(\eta>x)\d x+\int_{-\infty}^0\left(\nu(\eta>x)-1\right)\d x,~~\eta\in \mathcal{X}(\Theta,\mathcal{E}).$$ 
We say that $\eta,\zeta\in \mathcal{X}$ are \emph{comonotonic} if for all $s,s'\in \Theta,$
$(\eta(s)-\eta(s'))(\zeta(s)-\zeta(s'))\geq 0.$ 
As shown by \cite{S86}, a function $I:\mathcal{X}(\Theta,\mathcal{E})\to\mathbb{R}$ can be represented as $I(\eta)=\int_{\Omega}\eta \d\nu$ for some capacity $\nu$ if and only if $I$ is monotone ($I(\eta)\geq I(\zeta)$ whenever $\eta\geq \zeta$), comonotonic additive ($I(\eta+\zeta)=I(\eta)+I(\zeta)$ for comonotonic $\eta,\zeta$), and normalized ($I(1)=1$).

\subsection{The Savage setting}

\label{sec:sav}

Throughout the paper, we consider a measurable space $(\Omega, \mathcal{F})$ representing the set of possible future states of the world. An element of the set $\mathcal{X}(\Omega,\mathcal{F})$ is called an \emph{act} and, for notational convenience, we will simply denote this set as $\mathcal{X}$. Also for notational convenience, we will denote the set  $\mathcal{M}_1(\Omega, \mathcal{F})$, also called models, by $\Delta$. Given $\mu\in \Delta$ and $X,Y\in \mathcal{X}$, we say that $X\geq_{\mu}^{\mathrm{as}} Y$ if $\mu(X\geq Y)=1$.  Given $X\in \mathcal{X}$, we write $\|X\|=\sup_{\omega\in \Omega}|X(\omega)|$. A set $\mathcal{Y}\subseteq \mathcal{X}$ is said to be \emph{bounded} if $\sup_{X\in \mathcal{Y}}\|X\|<\infty$ and \emph{B-closed} if, for all bounded $(X_n)_{n\in N}\subseteq\mathcal{Y}$ satisfying $X_n\to X$ pointwise, it follows that $X\in \mathcal{Y}$.

Given a set $\mathfrak{S}$, a \emph{preference relation on $\mathfrak{S}$} is a total preorder; that is, (i) for all $\zeta,\eta\in \mathfrak{S}$, $\zeta\succsim \eta$ or $\eta\succsim \zeta$; (ii) for all $\zeta,\eta,\xi\in \mathcal{X}$, if $\zeta\succsim \eta$ and $\eta\succsim \xi$ then $\zeta\succsim \eta$. The Savage setting studies preference relations on $\mathcal{X}$ that represent a decision-maker's (DM) preferences. To ease later exposition, when we only say ``preference relation'', this is to be interpreted as a preference relation on $\mathcal{X}$.

In what follows, we introduce an axiom that will later be used in the axiomatization of a decision model under ambiguity in Section \ref{sec:Axioms}. We present it now because it will also be referenced in Section \ref{sec:two}.
\axiomplain{M}{For all $X,Y\in \mathcal{X}$, if $X\geq Y$ then $X\succsim Y$.}
Axiom (M) is standard for preference relations in the Savage setting. If the act $X$ pays more than the act $Y$ for every state of the world, then it is natural that the DM would prefer $X$ to $Y$.

\subsection{The pure-risk setting, distributions, and statistics}

\label{sec:risk}

Let $\mathcal{D}$ denote the set of compactly supported Borel probability measures on $\mathbb{R}$. We will refer to  $\mathcal{D}$ as the set of distributions. In decision theory, the study of preference relations on distributions is known as the pure-risk setting.

Many decision models in the Savage setting begin by reducing the problem to the pure-risk setting, e.g., the SEU model. This reduction is carried out as follows.
Suppose the DM does not perceive any ambiguity and describes their uncertainty with the unique probability measure $\mu\in \Delta$. Given any act $X\in\mathcal{X}$, the DM may represent this act with the distribution given by the pushforward probability measure $X_{\#}\mu$ defined by 
$$X_{\#}\mu(B)=\mu(X^{-1}(B)),~~B\in \mathcal{B}(\mathbb{R}),$$
where $\mathcal{B}(\mathbb{R})$ denotes the Borel-$\sigma$-algebra on $\mathbb{R}$ and $X^{-1}$ denote the preimage of $X$. This process allows the DM to move from the Savage setting to the pure-risk setting. Furthermore, if the DM has a preference relation on distributions $\succsim^{\ell}$, the DM is able to induce a preference relation $\succsim$ by
\begin{equation}
    \label{eq:ps}
    X\succsim Y\iff X_{\#}\mu \succsim^{\ell} Y_{\#}\mu,~~ X,Y\in \mathcal{X}.
\end{equation}
We call preference relations satisfying \eqref{eq:ps}  \emph{probabilistically sophisticated}, slightly relaxing the original definition of \cite{MS92} via stochastic dominance.

Given $Q\in \mathcal{D}$, the \emph{support of $Q$}, denote by $\supp(Q)$, is the smallest closed set $C\subseteq\mathbb{R}$ satisfying $Q(\mathbb{R}\backslash C)=0$. A collection $\mathcal{C}\subseteq \mathcal{D}$ is said to have \emph{uniformly bounded support} if there exists a compact interval that contains the support of each $Q\in \mathcal{C}$. Given $Q\in \mathcal{D}$, we define the \emph{survival function of $Q$} by $S_Q(x)=Q((x,\infty))$ for all $x\in \mathbb{R}$ and the \emph{quantile function of $Q$}
by $$q_Q(\alpha)=\inf\{x\in \mathbb{R}:S_Q(x)< 1-\alpha\},~~\alpha\in (0,1);$$ 
the mapping
$\VaR_\alpha: Q\mapsto q_Q(\alpha) $ is  called the \emph{Value-at-Risk of $Q$} at level $\alpha$ in finance.
Given an increasing (always in the non-strict sense) function $f:\mathbb{R}\to\mathbb{R}$ and $Q\in \mathcal{D}$, we define the distribution $f_{\#}Q$ by
$$f_{\#}Q(B)=Q(f^{-1}(B)),~~B\in \mathcal{B}(\mathbb{R}),$$
where $f^{-1}$ denote the preimage of $f$. It holds that $f_{\#}Q\in \mathcal{D}$ since $f$ is increasing. Given $(Q_n)_{n\in \N}\subseteq \mathcal{D}$ and $Q\in \mathcal{D}$, we say  $Q_n\to Q$ \emph{in distribution} if $\lim_{n\to\infty}\int_{\mathbb{R}} f\d Q_n= \int_{\mathbb{R}} f\d Q$ for all continuous and bounded $f:\mathbb{R}\to\mathbb{R}$.

For $Q,P\in \mathcal{D}$, we write  $Q\geq_{\mathrm{fsd}}P$ if $q_Q(\alpha)\geq q_P(\alpha)$ for all $\alpha\in (0,1)$, $Q\geq_{\mathrm{ssd}}P$ if
    \begin{equation}
        \int_{\mathbb{R}}f \d Q \geq \int_{\mathbb{R}}f \d P\label{eq:partial}
    \end{equation}
for all increasing concave functions $f:\mathbb{R}\to\mathbb{R}$, and 
$Q\geq_{\mathrm{cv}}P$ if \eqref{eq:partial} holds for all concave functions $f:\mathbb{R}\to\mathbb{R}$. Each of $\geq_{\mathrm{fsd}}$, $\geq_{\mathrm{ssd}}$, and $\geq_{\mathrm{cv}}$ defines a partial order on $\mathcal{D}$. Given $\mu\in \Delta$ and a partial order $\geq$ on $\mathcal{D}$, we can induce a partial order $\geq^{\mu}$ on $\mathcal{X}$ by
$$X\geq^{\mu}Y\iff X_{\#}\mu\geq Y_{\#}\mu,~~X,Y\in \mathcal{X}.$$
For a partial order $\geq$ on $\mathcal{D}$,  $\mathcal{C}\subseteq \mathcal{D}$, and $Q\in \mathcal{D}$, we say $\mathcal{C}\geq Q$ if for all $P\in \mathcal{C}$, $P\geq Q$. For $\mathcal{C}\subseteq \mathcal{D}$, the \emph{first-order infimum of $\mathcal{C}$}, denoted by $\bigwedge_{\mathrm{fsd}}\mathcal{C}$, is defined as the distribution $Q$ satisfying $P\geq_{\mathrm{fsd}}Q\geq_{\mathrm{fsd}}P'$ for all $P\in \mathcal{C}$ and $P'\in \mathcal{D}$ satisfying $\mathcal{C}\geq_{\mathrm{fsd}} P'$. Remark that if $Q$ exists, it is unique. If $\mathcal{C}\subseteq \mathcal{D}$ has uniformly bounded support, then $\bigwedge_{\mathrm{fsd}}\mathcal{C}$ exists (see e.g., \citealp{MWW25}).

A \emph{statistic} is a mapping $\gamma:\mathcal{D}\to\mathbb{R}$ such that for all $Q,P\in \mathcal{D}$,
    $$Q\geq_{\mathrm{fsd}}P ~(\text{resp.}~ Q>_\mathrm{fsd}P)\implies \gamma(Q)\geq \gamma(P)~(\text{resp.}~ \gamma(Q)>\gamma(P)),$$
    and for all sequences $(Q_n)_{n\in \N}\subseteq \mathcal{D}$ with uniformly bounded support converging in distribution to some $Q\in \mathcal{D}$, it holds that
    $\gamma(Q_n)\to\gamma(Q).$ A statistic $\gamma$ is said to be a \emph{certainty-equivalent statistic} if $\gamma(\delta_c)=c$ for all $c\in \mathbb{R}$, where $\delta_c$ denotes the Dirac measure at $c$.
    Statistics can be used to represent preference relations over distributions. That is, given the statistic $\gamma$, we can define a preference relation on distributions $\succsim^{\ell}$ by 
    \begin{equation}
        \label{eq:repLot}
        Q\succsim^{\ell}P\iff \gamma(Q)\geq \gamma(P),~~Q,P\in \mathcal{D}.
    \end{equation}
    \begin{proposition}\label{prop:riskfun}
    For a given $\succsim^{\ell}$ in \eqref{eq:repLot},  
        the statistic $\gamma$ is ordinally unique, 
        and there exists a unique certainty-equivalent statistic $\gamma$ satisfying \eqref{eq:repLot}.
    \end{proposition}
    
    The two defining properties of a statistic are natural, being nothing more than monotonicity and a mild continuity property, called compact continuity, studied by \cite{CM95}. 

    A \emph{von Neumann--Morgenstern (vNM) utility function} is a strictly increasing continuous function $u:\mathbb{R}\to\mathbb{R}$. A \emph{distortion function} is a strictly  increasing continuous function $g:[0,1]\to[0,1]$ with $g(0)=0$ and $g(1)=1$. The following are examples of statistics commonly used in decision theory. (a) The expected value statistic
    $$\mathbb{E}(Q)=\int_{\mathbb{R}}xQ(\d x),~~Q\in \mathcal{D}.$$ (b) The expected utility model of \cite{vNM47} corresponds to the statistic
    $$\gamma_{u}^{\mathrm{EU}}(Q)=\int_{\mathbb{R}}u\d Q=\mathbb{E}(u_{\#}Q),~~Q\in \mathcal{D},$$
    where $u$ is a vNM utility function. (c) The dual utility model of \cite{Y87} corresponds to the statistic
    $$\gamma_{g}^{\mathrm{DU}}(Q)=\int_0^\infty g(S_Q(x))\d x+\int_{-\infty}^0 (g(S_Q(x))-1)\d x,~~Q\in \mathcal{D},$$
    where $g$ is a distortion function. (d) The rank-dependent utility model of \cite{Q82} corresponds to the statistic
    $$\gamma_{u,g}^{\mathrm{RDU}}(Q)=\gamma^{\mathrm{DU}}_{g}\left(u_{\#}Q\right),~~Q\in \mathcal{D},$$
    where $u$ is a vNM utility function and $g$ is a distortion function.
Note that for $\mu\in \Delta$ and $X\in \mathcal X$,   $$\begin{aligned}
        \gamma^{\mathrm{RDU}}_{u,g}(X_{\#}\mu) =\int_0^\infty g(\mu(u(X)>x))\d x+\int_{-\infty}^0 (g(\mu(u(X)>x))-1)\d x =\int_{\Omega}u(X)\d(g\circ \mu),
    \end{aligned}$$
    which is how rank-dependent utility models are usually formulated.

    \begin{remark}
        Since, in general, the above statistics are not continuous with respect to convergence in distribution, the weaker notion of compact continuity is essential to our study.
    \end{remark}

\section{Two frameworks for separating pure risk and ambiguity}

\label{sec:two}

Both frameworks we discuss provide preference relations that clearly distinguish between pure risk and ambiguity. 
In these frameworks, we will assume that the DM represents the risk component of their preferences with a preference relation on distributions $\succsim^{\ell}$. Additionally, we will assume that $\succsim^{\ell}$ can be represented by a statistic as in \eqref{eq:repLot}. Where the frameworks differ is how they define and incorporate the ambiguity component of the preferences. 

\subsection{The evaluate-then-aggregate framework}

\label{sec:EA}

We equip the set of models $\Delta$ with the weak* topology, defined as the coarsest topology under which every map $\mu \mapsto \mu(A)$ is continuous for all $A \in \mathcal{F}$. It is common in decision theory to equip the space of models with the weak* topology, and no familiarity with this topology is needed for the discussion that follows. We denote by $\Sigma$ the corresponding Borel $\sigma$-algebra. 

In this framework, to describe the ambiguity component of the preferences, we will use the formulation proposed in \cite{KMM05} to axiomatize the smooth ambiguity model. In this framework, the DM also has a preference relation $\succsim^{m}$ on $\mathcal{X}(\Delta,\Sigma)$, which is hereby called a preference relation on model-based acts. We will assume that $\succsim^{m}$ satisfies the following regularity properties.
\axiomplainb{M1}{For all $\eta,\zeta\in \mathcal{X}(\Delta,\Sigma)$, if $\eta\geq\zeta$ then $\eta\succsim^{m}\zeta$.}

\axiomplainc{M2}{For all $\eta\in \mathcal{X}(\Delta,\Sigma)$, $\{c\in \mathbb{R}:\eta\succsim^{m} c\}$ and $\{c\in \mathbb{R}:c\succsim^{m} \eta\}$ are closed.}

\axiomplaina{M3}{For $a,b\in \mathbb{R}$, if $a>b$ then $a\succ^{m} b.$ }
\begin{proposition}\label{prop:labelGood1}
    Under Properties (M1)--(M3), there exists $\rho:\mathcal{X}(\Delta,\Sigma)\to\mathbb{R}$ with $\rho(\eta)\geq \rho(\zeta)$ for all $\eta,\zeta\in \mathcal{X}(\Delta,\Sigma)$
    satisfying $\eta\geq\zeta$, and $\rho(a)>\rho(b)$ for all $a,b\in \mathbb{R}$ satisfying $a>b$, such that
    \begin{equation}\label{eq:amPr}
    \eta\succsim^{m}\zeta\iff \rho(\eta)\geq \rho(\zeta),~~\eta,\zeta\in \mathcal{X}(\Delta,\Sigma).
\end{equation}
\end{proposition}

Using the preference relation on distributions $\succsim^{\ell}$, acts can be naturally transformed into model-based acts, as described below. Start by fixing the unique certainty-equivalent statistic $\gamma$ that represents $\succsim^{\ell}$ (Proposition \ref{prop:riskfun}). Given $X\in \mathcal{X}$, consider the mapping $\mathfrak{R}_X:\Delta\to\mathcal{D}$ defined by
$$\mathfrak{R}_X(\mu)=X_{\#}\mu,~~\mu\in \Delta.$$
The mapping $\mathfrak{R}_X$ allows us to define the model-based act $\gamma\circ \mathfrak{R}_X$. We show in Appendix \ref{app:two} that $\gamma\circ \mathfrak{R}_X\in \mathcal{X}(\Delta,\Sigma)$ for each $X\in \mathcal{X}$. Intuitively, given $X\in \mathcal{X}$ and $\mu\in\Delta$, the value $\gamma\circ \mathfrak{R}_X(\mu)$ can be seen as an evaluation, determined by $\succsim^{\ell}$, of the distribution $X_{\#}\mu$.

To couple the risk component $\succsim^{\ell}$ and the ambiguity component $\succsim^{m}$ into a preference relation $\succsim$, we assume that the DM's preference relation satisfies 
$$X\succsim Y\iff \gamma\circ \mathfrak{R}_X\succsim^{m}\gamma\circ \mathfrak{R}_Y,~~X,Y\in \mathcal{X}.$$
Therefore, the representation functional $\rho$ in \eqref{eq:amPr} aggregates the evaluation maps as
\begin{equation}\label{eq:EA}
    X\succsim Y\iff \rho(\gamma\circ \mathfrak{R}_X)\geq \rho(\gamma\circ \mathfrak{R}_Y),~~X,Y\in \mathcal{X}.
\end{equation}
We refer to any preference relation $\succsim$ satisfying \eqref{eq:EA} as an \emph{evaluate-then-aggregate (ETA) preference relation}. It is clear that any ETA preference relation satisfies Axiom (M). Additionally, if we take 
$$\rho(\eta)=\eta(\mu),~~\eta\in \mathcal{X}(\Delta,\Sigma),$$
for some $\mu\in \Delta$, then the ETA preference relation in \eqref{eq:EA} is probabilistically sophisticated.

Many standard decision models under ambiguity
can be formulated as ETA preference relations, and we list four below. Most of these models were first developed and axiomatized in the setting of \cite{AA63}; here we translate them to the Savage setting. In each example, we use the preference relation on distributions derived from the statistic $\gamma_{u}^{\mathrm{EU}}$ 
for some vNM utility $u$, 
which corresponds to the expected utility model of \cite{vNM47}. Remark that the certainty-equivalent statistic associated with $\gamma_{u}^{\mathrm{EU}}$ is the functional $u^{-1}\circ\gamma_{u}^{\mathrm{EU}}$, where $u^{-1}:\im(u)\to\mathbb{R}$ is the inverse of $u$.
(a) The maxmin expected utility model
of \cite{GS89} corresponds to
$$\rho(\eta)=\inf_{\mu\in \mathcal{Q}}u(\eta(\mu)),~~\eta\in \mathcal{X}(\Delta,\Sigma),$$
where $\mathcal{Q}\subseteq \Delta$. (b) The smooth ambiguity model of \cite{KMM05} corresponds to 
$$\rho(\eta)=\int_{\Delta}v(\eta)\d\pi,~~\eta\in \mathcal{X}(\Delta,\Sigma),$$
where $v$ is a vNM utility, and $\pi\in \mathcal{M}_1(\Delta,\Sigma)$. (c) The variational preferences model of \cite{MMR06} corresponds to
$$\rho(\eta)=\inf_{\mu\in \Delta}\left(u(\eta(\mu))+c(\mu)\right),~~\eta\in \mathcal{X}(\Delta,\Sigma),$$
where $c:\Delta\to[0,\infty]$ satisfies $c(\mu)<\infty$ for some $\mu\in \Delta$. (d) The $\alpha$-maxmin model of \cite{GMM04} corresponds to
$$\rho(\eta)=\alpha\inf_{\mu\in \mathcal{Q}}u(\eta(\mu))+(1-\alpha)\sup_{\mu\in \mathcal{Q}}u(\eta(\mu)),~~\eta\in \mathcal{X}(\Delta,\Sigma),$$
where $\mathcal{Q}\subseteq \Delta$ and $\alpha\in [0,1]$.

In the following section, we propose an alternative to the evaluate-then-aggregate framework.

\subsection{The aggregate-then-evaluate framework}

\label{sec:AE}

In the ATE framework, to describe the ambiguity component of the preferences, the DM fixes a mapping $\mathfrak{D}:\mathcal{X}\to\mathcal{D}$, which we call an \emph{act-to-distribution mapping}. Intuitively, the DM uses an act-to-distribution mapping to aggregate their ambiguity regarding an act into a single distribution. Given $\mu\in \Delta$, the \emph{probabilistic act-to-distribution mapping}, denoted by $\mathfrak{D}_{\mu}$, is defined by $\mathfrak{D}_{\mu}(X)=X_{\#}\mu$ for all $X\in \mathcal{X}.$ How much the DM perceives and reacts to ambiguity is determined by how much their act-to-distribution mapping $\mathfrak{D}$ differs from a probabilistic act-to-distribution mapping. 
\begin{example}
    \label{example:atd}
    Let $\mathcal{Q}\subseteq \Delta$, we can define the act-to-distribution mapping $\mathfrak{D}_{\mathcal{Q}}$ by
    $$\mathfrak{D}_{\mathcal{Q}}(X)=\bigwedge_{\mathrm{fsd}}X_{\#}\mathcal{Q},~~X\in \mathcal{X},$$
    where $X_{\#}\mathcal{Q}=\{X_{\#}\mu:\mu\in \mathcal{Q}\}$.
    To see this mapping is well-defined, given $X\in \mathcal{X}$, find $a,b\in \mathbb{R}$ such that $a\geq X\geq b$. For all $\mu \in \Delta$, since $\mu(a\geq X\geq b)=1$, we have $\supp(X_{\#}\mu)\subset [b,a]$. Therefore, the set $\mathcal{X}_{\#}\mathcal{Q}$ has uniformly bounded support and  $\bigwedge_{\mathrm{fsd}}X_{\#}\mathcal{Q}$ exists. 
\end{example}

The following are natural properties for act-to-distribution mappings, and they are all satisfied by probabilistic act-to-distribution mappings. The first two properties are straightforward and need no further motivation.
\axiomplainb{R1}{For all $X,Y\in \mathcal{X}$, if $X\geq Y$ then $\mathfrak{D}(X)\geq_{\mathrm{fsd}}\mathfrak{D}(Y)$.}

\axiomplaina{R2}{For all bounded $(X_n)_{n\in \N}\subseteq \mathcal{X}$, if $X_n\to X$ pointwise then $\mathfrak{D}(X_n)\to \mathfrak{D}(X)$ in distribution.}

To motivate the final property, consider a decision maker with a vNM utility function $u$ and an act-to-distribution mapping $\mathfrak{D}$. The function $u$ converts monetary outcomes into their utility scale that reflects how appealing each outcome is to the decision maker. Given an act $X \in \mathcal{X}$, there are two ways the decision maker can form a distribution on the utility scale. First, they can apply $\mathfrak{D}$ to the act to obtain a monetary distribution, that is, $\mathfrak{D}(X)$, and then apply $u$ to this distribution, that is, $u_{\#}(\mathfrak{D}(X))$. Second, they can apply $u$ directly to the act $X$ and then apply $\mathfrak{D}$ to the resulting utility-valued act, that is, $\mathfrak{D}(u(X))$. The following property says that these two procedures should obtain the same distribution.

\axiomplain{R3}{For all $X\in \mathcal{X}$ and vNM utility functions $u$, it holds that
    $\mathfrak{D}(u(X))=u_{\#}(\mathfrak{D}(X)).$}

To couple the risk component $\succsim^{\ell}$ and the ambiguity component $\mathfrak{D}$ into a preference relation $\succsim$, the DM starts by fixing a statistic $\gamma$ that represents $\succsim^{\ell}$. The preference relation $\succsim$ of the DM can be represented by
\begin{equation}\label{eq:AE}
    X\succsim Y\iff \gamma\circ\mathfrak{D}(X)\geq \gamma\circ\mathfrak{D}(Y),~~X,Y\in \mathcal{X}.
\end{equation}
We refer to any preference relation $\succsim$ given by \eqref{eq:AE} for any statistic $\gamma:\mathcal D\to \R$ and any act-to-distribution mapping $\mathfrak{D}:\X\to \mathcal D$ as an \emph{aggregate-then-evaluate  (ATE) preference relation}. It is clear that any preference relation $\succsim$ satisfying \eqref{eq:AE}, with $\mathfrak{D}$ satisfying Property (R1), must also satisfy Axiom (M). Additionally, if we take $\mathfrak{D}$ to be a probabilistic act-to-distribution mapping, then the preference relation $\succsim$ is probabilistically sophisticated. 

In what follows, we present a method for constructing act-to-distribution mappings inspired by \cite{SW94}. The key distinction is that our definition applies to all acts, rather than being restricted to those with a finite range.
Given a continuous capacity $\nu$, the \emph{Choquet act-to-distribution mapping}, denoted by $\mathfrak{D}_{\nu}$, is defined by
 $$S_{\mathfrak{D}_{\nu}(X)}(x)=\nu(X>x)~\text{for all}~x\in \mathbb{R}.$$ The mapping $\mathfrak{D}_{\nu}$ is well-defined because $\nu$ is continuous and survival functions uniquely characterize elements of $\mathcal{D}$. It is straightforward to show that if $\nu$ is a probability measure, then $\mathfrak{D}_\nu$ is a probabilistic act-to-distribution mapping. Also, by the well-known layer-cake representation, we have
 $$\begin{aligned}
     \mathbb{E}\circ \mathfrak{D}_\nu(X)&=\int_0^\infty S_{\mathfrak{D}_\nu(X)}(x)\d x+\int_{-\infty}^0 (S_{\mathfrak{D}_\nu(X)}(x)-1)\d x\\&=\int_0^\infty \nu(X>x)\d x+\int_{-\infty}^0 (\nu(X>x)-1)=\int_{\Omega} X\d\nu,~~X\in \mathcal{X}.
 \end{aligned}$$

 Preference relations $\succsim$ satisfying \eqref{eq:AE} with $\mathfrak{D}$ given by a Choquet act-to-distribution mapping are said to be \emph{Choquet ATE} preference relations. We have the following theorem characterizing all Choquet act-to-distribution mappings.

\begin{theorem}
    \label{theo:riskEq}
    Let $\mathfrak{D}$ be an act-to-distribution mapping. Then $\mathfrak{D}$ satisfies Properties (R1)--(R3) if and only if $\mathfrak{D}$ is a Choquet act-to-distribution mapping.
\end{theorem}

Theorem \ref{theo:riskEq} is related to \citet[Theorem 1]{LSW21}, which characterizes transformations $T:\mathcal{D}\to\mathcal{D}$   satisfying
$$S_{T(Q)} = g \circ S_Q,~~ Q \in \mathcal{D} $$
for a distortion function $g$, with a property similar to (R3). The main difference between the two results is that we consider mappings defined on acts, whereas \citet{LSW21} studied transformations defined on distributions.

Given a continuous capacity $\nu$ and a vNM utility $u$, by Theorem \ref{theo:riskEq}, we have 
 $$\gamma^{\mathrm{EU}}_u\circ\mathfrak{D}_{\nu}(X)=\mathbb{E}(u_{\#}\mathfrak{D}_{\nu}(X))=\mathbb{E}\circ\mathfrak{D}_{\nu}(u(X))=\int_{\Omega}u(X)\d\nu,~~X\in \mathcal{X}.$$ Therefore, preference relations given by the Choquet expected utility model of \cite{S89}  with a continuous capacity  are examples of Choquet ATE preference relations.

\section{Axioms and representation for the Choquet ATE model}

\label{sec:Axioms}

We next state and discuss several additional axioms for preference relations $\succsim$, which 
lead to a main representation result in 
Theorem \ref{theo:mainrep}. 

For the rest of this paper, we will work in a Savage setting with  an exogenous source of unambiguous events.
The source of unambiguous events is modeled by a fixed sub-$\sigma$-algebra $\mathcal{G}$. 
Since the events in $\mathcal{G}$ are unambiguous, the DM can specify a probability measure $\mathbb{P}\in \mathcal{M}_1(\Omega,\mathcal{G})$ 
that assigns likelihoods to the unambiguous events. As is common in risk management, we will assume that $\mathbb{P}$ is atomless; that is, $\{\mathbb{P}(B) : B \in \mathcal{G},~B \subseteq A\}=[0, \mathbb{P}(A)] $ for every $A \in \mathcal{G}$. We will refer to this setting as the \emph{Savage setting with pure risk}. We call $\mathbb{P}$ the \emph{exogenous probability measure}.

We write $\mathcal{X}(\mathcal{G})$ for the subspace of $\mathcal{X}$ consisting of $\mathcal{G}$-measurable acts. Under our setting, since acts $X\in \mathcal{X}(\mathcal{G})$ exhibit no distributional ambiguity, the preference relation $\succsim$ restricted to $\mathcal{X}(\mathcal{G})$ should be determined solely by the distributions of these acts under the exogenous probability measure $\mathbb{P}$. This brings us to our second axiom of risk conformity. 
\axiomplain{RC}{Risk conformity: For all $X,Y\in \mathcal{X}(\mathcal{G})$ satisfying $X=^{\mathbb{P}}_{\mathrm{fsd}}Y$, $X\simeq Y$.}
For further discussions of Axiom (RC) in general decision models, see \cite{BBLW25}, and for its role in risk assessment, see \cite{SVW25}. If the preference relation $\succsim$ satisfies Axiom (RC), then the restriction of $\succsim$ to $\mathcal{X}(\mathcal{G})$, called the \emph{pure-risk preferences} of $\succsim$, corresponds to a preference relation over distributions. 

We say that an act-to-distribution mapping $\mathfrak{D}$ is \emph{risk conforming} if $\mathfrak{D}(X)=X_{\#}\mathbb{P}$ for all $X\in \mathcal{X}(\mathcal{G}).$ Let $\succsim$ be an ATE preference relation. In the Savage setting with pure risk, to remain consistent with the interpretation of an act-to-distribution mapping, $\mathfrak{D}$ in \eqref{eq:AE} should be risk conforming. This follows as the $\mathcal{G}$-measurable acts have no distributional ambiguity. A capacity $\nu$
is \emph{risk conforming} if $\nu|_{\mathcal{G}}=\mathbb{P}$.
In this case, risk conformity means that, for unambiguous events, the capacity coincides with the exogenous probability measure $\p$. When $\nu$ describes the DM's subjective evaluation regarding the likelihoods of events, risk conformity is essential, as the DM believes that $\mathbb{P}$ correctly specifies the likelihood of unambiguous events. Let $\nu$ be a continuous capacity. It is straightforward to show that $\nu$ is risk conforming if and only if the Choquet act-to-distribution mapping $\mathfrak{D}_\nu$ is risk conforming.

The next two axioms are regularity requirements.  
\axiomplainb{SRM}{Strict risk monotonicity: For all $X,Y\in \mathcal{X}(\mathcal{G})$ satisfying $X>^{\mathbb{P}}_{\mathrm{as}} Y$, $X\succ Y$.}

\axiomplaina{C}{Continuity: For all $X\in \mathcal{X}$, $\{Y\in \mathcal{X}:X\succsim Y\}$ and $\{Y\in \mathcal{X}:Y\succsim X\}$ are B-closed.}

Axiom (SRM) reflects the fact that the DM believes that the probability measure $\mathbb{P}$ is well-specified on $\mathcal{G}$. Thus, a stronger form of monotonicity, using distributions, holds for acts in $\mathcal{X}(\mathcal{G})$. 
Axiom (C) is a reformulation of the compact continuity in \cite{CM95} on acts, and it is commonly used in the literature on risk measures; see \cite{FS16}.

The axioms discussed thus far concern only certain notions of regularity and essentially impose no specific functional forms on the preference relations. These regularity axioms imply the existence of matching probabilities, as defined below. For notational simplicity, we will write $A\succsim B$ for $A,B\in \mathcal{F}$ to mean $\id_A\succsim\id_B$, where $\id_A$ denotes the binary act which yields $1$ if $\omega\in A$ and $0$ otherwise.
\begin{definition}
Given a preference relation $\succsim$, a function $\nu:\mathcal{F}\to [0,1]$ is a \emph{$\succsim$-matching probability} if for all $A\in \mathcal{F}$, there exists $R_A\in \mathcal{G}$ such that $\nu(A)=\mathbb{P}(R_A)$ and $A\simeq R_A$.    
\end{definition}

Intuitively, a $\succsim$-matching probability uses the restriction of the preference relation $\succsim$ on $\{\id_A: A\in \mathcal{F}\}$ to describe  ambiguity on the level of events, by equalizing each ambiguous event to an unambiguous one.  This idea requires the Savage setting with pure risk, as $\mathcal{G}$ needs to contain the unambiguous events.  The following result is taken from \citet[Proposition 2]{VW25}.

\begin{proposition}
    \label{lemma:match}
    For every preference relation $\succsim$ satisfying Axioms (M), (RC), (SRM), and (C), there exists a unique $\succsim$-matching probability $\nu$. Furthermore, $\nu$ is a continuous risk-conforming capacity, and 
    $$A\succsim B\iff \nu(A)\geq \nu(B),~~A,B\in \mathcal{F}.$$
\end{proposition}

As we will see in Theorem \ref{theo:mainrep}, matching probabilities will fully determine the ambiguity component of the preferences. To discuss the risk component of the preferences, we have that these regularity axioms are sufficient to represent the pure-risk preferences via a statistic.

\begin{proposition}
    \label{prop:need1}
    If the preference relation $\succsim$ satisfies Axioms (M), (RC), (SRM), and (C), then there exists a statistic $\gamma$ such that
    \begin{equation}\label{eq:risk}
        X\succsim Y \iff \gamma(X_{\#}\mathbb{P})\geq \gamma(Y_{\#}\mathbb{P}),~~X,Y\in \mathcal{X}(\mathcal{G}).
    \end{equation}
\end{proposition}

By Proposition \ref{prop:riskfun}, there exists a unique certainty-equivalent statistic $\gamma$ that satisfies \eqref{eq:risk}. For any preference relation $\succsim$ satisfying Axioms (M), (RC), (SRM), and (C), we refer to this unique certainty-equivalent statistic as the \emph{$\succsim$-CES}.

The final axiom is from \cite{SW94} and is used to couple the risk and ambiguity components of the preferences. Before introducing it, we require some preliminary definitions. Given a preference relation $\succsim$, we can define the survival relation $\succsim^*$ given $$X\succsim^* Y\iff \{X>x\}\succsim\{Y>x\} ~\text{for all}~x\in\mathbb{R},~~X,Y\in \mathcal{X}.$$
This allows us to define our final axiom.
\axiomplain{CD}{Cumulative dominance: For all $X,Y\in \mathcal{X}$, if $X\simeq^*Y$ then $X\simeq Y$.}
The following result shows that the above axioms fully characterize Choquet ATE preference relations in the Savage setting with pure risk.

\begin{theorem}
    \label{theo:mainrep}
    The preference relation $\succsim$ satisfies Axioms (M), (RC), (SRM), (C), and (CD) if and only if $\succsim$ is a Choquet ATE preference relation, that is, there exists a statistic $\gamma$ and a risk-conforming capacity $\nu$ such that
    \begin{equation}
        \label{eq:mainrep}
        X\succsim Y\iff \gamma\circ\mathfrak{D}_{\nu}(X)\geq \gamma\circ\mathfrak{D}_{\nu}(Y).
    \end{equation}
    Moreover, in \eqref{eq:mainrep}, the statistics $\gamma$ is ordinally unique, and the capacity $\nu$ is unique and coincides with the $\succsim$-matching probability.
\end{theorem}

Let $\succsim$ be a Choquet ATE preference relation in \eqref{eq:mainrep}. Assume that $\gamma$ is given by the statistic $\gamma^{\mathrm{RDU}}_{u,g}$ corresponding to the rank-dependent utility model. We have
$$\gamma^{\mathrm{RDU}}_{u,g}\circ\mathfrak{D}_{\nu}(X)=\gamma^{\mathrm{DU}}_g(u_{\#}\mathfrak{D}_{\nu}(X))=\gamma^{\mathrm{DU}}_g(\mathfrak{D}_{\nu}(u(X)))=\int_{\Omega}u(X)\d(g\circ \nu),~~X\in \mathcal{X}.$$
Therefore, as $\nu$ is risk conforming, we have that $\succsim$ is an example of a Choquet rank-dependent utility preference relation, see \cite{TF95}. This observation will motivate Section \ref{sec:AA}, which discusses ambiguity attitudes.

\begin{example}
    Allowing general statistics in Choquet ATE preference relations leads to new decision models under ambiguity. For example, given $a\in \mathbb{R}\backslash\{0\}$, let $u_a$ be the vNM utility function defined by $u_a(x)=e^{ax}$. For $a\in \mathbb{R}$, define the certainty-equivalent statistic $\gamma_a$ by 
    $$\gamma_a(Q)=\frac{1}{a}\log\left(\gamma_{u_a}^\mathrm{EU}(Q)\right)~\text{if}~a\in \mathbb{R}\backslash\{0\},~\text{and}~\gamma_0(Q)=\mathbb{E}(Q),~~Q\in \mathcal{D}.$$
    The statistic $\gamma_a$ is the entropic risk measure (see \citealp{FS02}). Given a Borel probability measure $\kappa$ on $\mathbb{R}$, \citet{MPST24}  studied a statistic given by
    $$\gamma_{\kappa}(Q)=\int_{\mathbb{R}}\gamma_a(Q)\kappa(\d a),~~Q\in \mathcal{D}.$$
    If $\succsim$ is Choquet ATE preference relation with the $\succsim$-CES   given by $\gamma_{\kappa}$, then  
    \begin{equation}
        \label{eq:amE}
        X\succsim Y\iff \int_{\mathbb{R}}\frac{1}{a}\log\left(\int_\Omega e^{aX}\d\nu\right)\kappa(\d a)\geq \int_{\mathbb{R}}\frac{1}{a}\log\left(\int_\Omega e^{aY}\d\nu\right)\kappa(\d a), ~~X,Y\in \mathcal{X},
    \end{equation}
    where $\nu$ is the $\succsim$-matching probability. The model \eqref{eq:amE}  has a natural interpretation. To account for ambiguity, the DM first evaluates each act using the certainty equivalent under the Choquet expected utility model with utility function $u_a$, for every $a\in \mathbb{R}$. These values are then averaged with respect to the measure $\kappa$. An axiomatization of this model can be obtained by taking the axioms of Theorem \ref{theo:mainrep} and adding the following axiom: for all $X,Y,Z \in \mathcal{X}(\mathcal{G})$ with $Z$ independent of both $X$ and $Y$ under the exogenous probability measure, we have 
    $$X\succsim Y\iff X+Z\succsim Y+Z.$$ This follows directly from \citet[Theorem 1]{MPST24}, which characterizes the class of statistics  $\gamma_\kappa$.
\end{example}

\section{Risk and ambiguity attitudes}
In this section, we study how attitudes toward pure risk and ambiguity can be separated within the ATE framework. We begin by isolating the risk component of the preferences and recalling several classical risk attitudes. We then introduce comparative and absolute ambiguity attitudes and conclude by analyzing how risk and ambiguity attitudes can lead to   preferences for diversification.

\label{sec:RA}

\subsection{Risk attitudes}

\label{sec:RAttitudes}

In this section, we will call the Axioms (M), (RC), (SRM), and (C) the regularity axioms. 
Let $\succsim$ be a preference relation satisfying the regularity axioms. 
Risk attitudes for $\succsim$ correspond to properties of the pure-risk preferences associated with $\succsim$. Intuitively, risk attitudes should be represented by properties of the $\succsim$-CES, which isolates the risk component of $\succsim$. An independent study of risk attitudes is necessary because their interaction with ambiguity attitudes (properties of the act-to-distribution mapping) plays a crucial role in determining the properties of the overall preference relation, as we will see in Section \ref{sec:div}. Additionally, this section is important because the pure-risk preferences of preference relations satisfying the regularity axioms coincide with the preference relations in the pure-risk setting studied in \cite{CL07}, enabling us to draw directly on their results while also establishing new ones. We now recall several standard risk attitudes commonly discussed in the literature.

The first risk attitude we will discuss is the strong risk aversion as studied by \cite{RS70}.
\axiomplain{SRA}{Strong risk aversion:
    For all $X,Y\in \mathcal{X}(\mathcal{G})$, if $X\geq_{\mathrm{ssd}}^{\mathbb{P}} Y$ then $X\succsim Y$.}
Given a preference relation $\succsim$ satisfying the regularity axioms with the $\succsim$-CES denoted by $\gamma$, it is clear that $\succsim$ satisfies Property (SRA) if and only if $\gamma$ is \emph{monotonic with respect to $\geq_{\mathrm{ssd}}$}, i.e., for all $Q,P\in \mathcal{D}$ satisfying $Q\geq_{\mathrm{ssd}}P$,
$\gamma(Q)\geq \gamma(P)$. For example, a direct consequence of \citet[Corollary 3]{SZ08} is that the statistic $\gamma^{\mathrm{RDU}}_{u,g}$ is monotonic with respect to $\geq_{\mathrm{ssd}}$ if and only if $u$ is concave and $g$ is convex.

The next risk attitude we will discuss is risk diversification, which is reformulated from the main property in \cite{D89}. 
\axiomplain{RD}{Risk diversification: For all $X,Y\in \mathcal{X}(\mathcal{G})$ and $\lambda\in [0,1]$, if $X\simeq Y$ then  $\lambda X+(1-\lambda)Y\succsim X.$}
If the preference relation of the DM satisfies Property (RD), then the DM prefers mixtures of equally preferred unambiguous acts. This is consistent with the intuitive notion of diversification. We will discuss the overall preference for diversification in Section \ref{sec:div}.

Given a preference relation $\succsim$ satisfying the regularity axioms, our goal is to characterize Property (RD) through properties of the associated $\succsim$-CES. An important first step in this direction is provided by \citet[Theorem 3]{CL07}, who show that Property (SRA) is equivalent to the following weaker form of Property (RD): For all $X,Y\in \mathcal{X}(\mathcal{G})$ and $\lambda\in [0,1]$, it holds that
 $$
    X=^{\mathbb{P}}_{\mathrm{fsd}}Y \implies \lambda X+(1-\lambda)Y\succsim X.
 $$
Therefore, for a preference relation satisfying the regularity axioms, Property (RD) implies Property (SRA). However, as shown by \cite{D89}, the two properties are not equivalent in general. 

Since Property (RD) is strictly stronger than Property (SRA), it is natural to ask what properties of the $\succsim$-CES, additional to monotonicity with respect to $\geq_{\mathrm{ssd}}$, are needed for Property (RD) to hold. A compelling candidate is a quasiconcavity property. Therefore, we must first introduce an appropriate mixture operation on the set of distributions $\mathcal{D}$. 

Given $Q,P\in \mathcal{D}$, denote by $Q\oplus P$ the element of $\mathcal{D}$ satisfying 
$$\mathrm{VaR}_{\alpha}(Q\oplus P)=\mathrm{VaR}_{\alpha}(Q)+\mathrm{VaR}_{\alpha}(P)~\text{for all}~\alpha\in (0,1).$$
Given $Q\in \mathcal{D}$ and non-negative $a\in \mathbb{R}$ denote by $a\otimes Q$ the element of $\mathcal{D}$ satisfying
$$\mathrm{VaR}_{\alpha}(a\otimes Q)=a \mathrm{VaR}_{\alpha}(Q)~\text{for all}~\alpha\in(0,1).$$
As discussed in \citet[Proposition 3]{LSW21}, the triplet $(\mathcal{D},\oplus,\otimes)$ is a convex cone. The following proposition demonstrates that the above operations naturally appear when considering Choquet act-to-distribution mappings.

\begin{proposition}
    \label{prop:oper}
    Let $\nu$ be a continuous capacity. For comonotonic $X,Y\in \mathcal{X}$ and $\lambda\in [0,1]$, we have
    $$\mathfrak{D}_\nu(X+Y)=\mathfrak{D}_\nu(X)\oplus \mathfrak{D}_\nu(Y)~~\text{and}~~\mathfrak{D}_\nu(\lambda X)=\lambda \otimes \mathfrak{D}_\nu(X).$$
\end{proposition}

We say that a statistic $\gamma$ is \emph{comonotonic quasiconcave} if for all $Q, P\in \mathcal{D}$ and $\lambda\in [0,1]$,
$$\gamma\left(\lambda \otimes Q\oplus (1-\lambda)\otimes P\right)\geq \min\{\gamma(Q),\gamma(P)\}.\footnote{Here $\otimes$ takes precedence over $\oplus$ in the order of operations.}$$
A consequence of \citet[Theorem 3]{WY19} is that the statistic $\gamma^{\mathrm{RDU}}_{u,g}$ is comonotonic quasiconcave if and only if $u$ is concave. Therefore, if $\gamma^{\mathrm{RDU}}_{u,g}$ satisfies monotonicity with respect to $\geq_{\mathrm{ssd}}$, then $\gamma^{\mathrm{RDU}}_{u,g}$ is comonotonic quasiconcave. Therefore, in the case of the statistic $\gamma^{\mathrm{RDU}}_{u,g}$, monotonicity with respect to $\geq_{\mathrm{ssd}}$ implies comonotonic quasiconcavity.

Consider the following risk attitude studied in \cite{CT02}, called comonotonic diversification there.
\axiomplain{WRD}{Weak risk diversification:
    For all comonotonic $X,Y\in \mathcal{X}(\mathcal{G})$ and $\lambda\in [0,1]$, 
    $$X\simeq Y\implies \lambda X+(1-\lambda)Y\succsim X.$$}
It is clear that Property (WRD) is weaker than Property (RD). As it turns out, comonotonic quasiconcavity of the $\succsim$-CES characterizes the Property (WRD).

\begin{proposition}
    \label{prop:coDi}
    Let $\succsim$ be a preference relation satisfying the regularity axioms. Then $\succsim$ satisfies Property (WRD) if and only if 
    the $\succsim$-CES is comonotonic quasiconcave.
\end{proposition}
Therefore, if a preference relation $\succsim$ satisfying the regularity axioms additionally satisfies Property (RD), then the $\succsim$-CES is comonotonic quasiconcave. These observations lead us to the main result
of the section.

\begin{theorem}\label{theo:rD}
    Let $\succsim$ be a preference relation satisfying the regularity axioms. Then $\succsim$ satisfies Property (RD) if and only if $\succsim$ satisfies Properties (SRA) and (WRD). 
\end{theorem}
For a preference relation $\succsim$ satisfying the regularity axioms, Theorem \ref{theo:rD} implies that  $\succsim$ satisfies Property (RD) if and only if the $\succsim$-CES satisfies monotonicity with respect to $\geq_{\mathrm{ssd}}$ and comonotonic quasiconcavity. In \cite{D89}, it was determined that a quasiconcavity condition in combination with monotonicity with respect to $\geq_{\mathrm{ssd}}$ was sufficient to imply Property (RD). However, the probabilistic quasiconcavity in \cite{D89} could only provide sufficiency, whereas our comonotonic quasiconcavity provides an equivalence.

We conclude this section by discussing when risk attitudes extend to properties of the overall preference relation in the case of Choquet ATE preference relations. The first example of this phenomenon is regarding Property (WRD).

\begin{proposition}
    \label{prop:Extend1}
    Let $\succsim$ be a Choquet ATE preference relation. If $\succsim$ satisfies Property (WRD) then for all comonotonic $X,Y\in \mathcal{X}$ and $\lambda\in [0,1]$, $X\simeq Y\implies \lambda X+(1-\lambda)Y\succsim X.$
\end{proposition}

The remaining properties pertain to statistics $\gamma$. Nonetheless, it is easy to see that these properties could just as well be formulated in terms of pure-risk preferences. We say that  a  statistic $\gamma$  is \emph{constant additive} if for all $Q\in \mathcal{D}$ and $a\in \mathbb{R}$, $\gamma(Q\oplus\delta_a)=\gamma(Q)+a$; \emph{positively homogeneous} if for all $Q\in \mathcal{D}$ and non-negative $a\in \mathbb{R}$, $\gamma(a\otimes X)=a\gamma(Q)$; biseparable if there exists a vNM utility $u$ and a statistic $\tilde{\gamma}$ satisfying constant additivity and positive homogeneity such that
    \begin{equation}
    \label{eq:bis}
        \gamma(Q)=\tilde{\gamma}(u_{\#}Q),~~Q\in \mathcal{D}.
    \end{equation}
Constant additivity and positive homogeneity are common properties discussed in the literature of risk measures; see \cite{FS16}. Biseparability is a property of biseparable preference relations. Traditionally, they are considered one of the weakest classes of preference relations, encompassing many famous decision models under ambiguity. For more on biseparable preference relations, see \cite{GM01}.

\begin{proposition}\label{prop:riskEx}
    Let $\succsim$ be a Choquet ATE preference relation.
    \begin{enumerate}[(i) ]
        \item If the $\succsim$-CES satisfies constant additivity, then for all $a\in \mathbb{R}$ we have
        $$X\succsim Y\iff X+a\succsim Y+a,~~X,Y\in \mathcal{X}.$$
        \item If the $\succsim$-CES satisfies constant additivity, then for all positive $a\in \mathbb{R}$ we have
        $$X\succsim Y\iff aX\succsim aY,~~X,Y\in \mathcal{X}.$$
        \item If the $\succsim$-CES satisfies biseparability, then there exists a vNM function $u$ and  $I:\mathcal{X}\to\mathbb{R}$ satisfying $I(X)\geq I(Y)$ for all $X,Y\in\mathcal{X}$ with $X\geq Y$ and $I(aX+b)=aI(X)+b$ for all $X\in \mathcal{X}$, non-negative $a\in \mathbb{R}$, and $b\in \mathbb{R}$, such that
        $$X\succsim Y\iff I(u(X))\geq I(u(Y)), ~~X,Y\in \mathcal{X}.$$
    \end{enumerate}
\end{proposition}
A key consequence of Proposition \ref{prop:riskEx} is that, for a Choquet ATE preference relation $\succsim$, the $\succsim$-CES being biseparable implies that the preference relation $\succsim$ is biseparable in the sense of \cite{GM01}.

\subsection{Comparative and absolute ambiguity attitudes}

\label{sec:AA}

We next define comparative ambiguity attitudes for Choquet ATE preference relations. This will then allow us to define absolute ambiguity attitudes once a normalization for ambiguity neutrality is provided. 

Given the Choquet ATE preference relations $\succsim_1$ and $\succsim_2$, we say that $\succsim_2$ is \emph{more ambiguity averse than} $\succsim_1$ if 
$$R\succsim_1 A \implies R\succsim_2 A,~~~~~~R\in \mathcal{G}~\text{and}~A\in \mathcal{F}.$$
The interpretation of this definition is quite natural: As events $R\in \mathcal G$ have no ambiguity, they can serve as the benchmark for comparing ambiguous events. 
This is similar to the classic notion of comparative risk aversion (\citealp{P64,Y69}), where constant acts serve as the benchmark. 

We have the following proposition characterizing comparative ambiguity attitudes. 
\begin{theorem}
    \label{prop:comAm}
    Let $\succsim_1$ and $\succsim_2$ be Choquet ATE preference relations with matching probabilities $\nu_1$ and $\nu_2$ respectively. Then, $\succsim_2$ is more ambiguity averse than $\succsim_1$ if and only if 
    $\nu_1(A)\geq \nu_2(A)$ for all $A\in \mathcal{F}$.
\end{theorem}

In Theorem \ref{prop:comAm}, the $\succsim_1$-CES and the $\succsim_2$-CES may not be the same.
This is consistent with the fact that the ambiguity component in the Choquet ATE model is separated from the risk component. 
To get a full comparative implication on all acts, instead of only binary ones, we need to force the risk components of both preference relations to be identical, as shown in the next result. 

\begin{proposition}
    \label{prop:attitudes}
    Let $\succsim_1$ and $\succsim_2$ be Choquet ATE preference relations. Then, the comparative implication  \begin{equation}
        \label{eq:ep}
        X\succsim_1 Y~(X\succ_1 Y) \implies X\succsim_2 Y~(X\succ_2 Y),~~\text{for all}~X\in \mathcal{X}(\mathcal{G})~\text{and}~Y\in \mathcal{X}
    \end{equation} holds if and only if the $\succsim_1$-CES is equal to the $\succsim_2$-CES and $\succsim_2$ is more ambiguity averse than $\succsim_1$.
\end{proposition}

As Proposition \ref{prop:attitudes}  demonstrates, our notion of comparative ambiguity attitudes for Choquet ATE preference relations generalizes the comparative ambiguity attitudes of \cite{E04} given by \eqref{eq:ep}.

Theorem \ref{prop:comAm} suggests a natural way to define comparative ambiguity attitudes for general act-to-distribution mappings. To see this, observe that for two continuous capacities $\nu_1$ and $\nu_2$, the condition $\nu_1(A)\geq \nu_2(A)$ for all $A\in \mathcal{F}$ is equivalent to the requirement that $\mathfrak{D}_{\nu_1}(X)\geq_{\mathrm{fsd}}\mathfrak{D}_{\nu_2}(X)$ for all $X\in \mathcal{X}$. Therefore, given two act-to-distribution mappings $\mathfrak{D}_1$ and $\mathfrak{D}_2$, we say that \emph{$\mathfrak{D}_2$ is more ambiguity averse than $\mathfrak{D}_1$} if 
$$\mathfrak{D}_1(X)\geq_{\mathrm{fsd}}\mathfrak{D}_2(X),~~X\in \mathcal{X}.$$

The next step is to define a normalization for ambiguity neutrality. 
\axiomplain{AN}{Ambiguity neutrality: For all $A,B,C\in \mathcal{F}$ satisfying $(A\cup B)\cap C=\varnothing$,
    $$A\succsim B\iff A\cup C\succsim B\cup C.$$}
 
The following proposition is taken from \citet[Theorem 2]{VW25}, which characterizes Property (AN).

\begin{proposition}\label{prop:AN}
    Let $\succsim$ satisfy Axioms (RC), (M), (SRM), and (C). Then $\succsim$ satisfies Property (AN) if and only if the $\succsim$-matching probability $\nu$ is a probability measure.
\end{proposition}

Finally, we say that a Choquet ATE preference relation $\succsim$ is \emph{ambiguity averse} if it is more ambiguity averse than a Choquet ATE preference relation satisfying Property (AN). A straightforward consequence of Theorem \ref{prop:comAm} and Proposition \ref{prop:AN} is that given a Choquet ATE preference relation $\succsim$, ambiguity aversion is equivalent to the $\succsim$-matching probability being balanced.

Once again, Proposition \ref{prop:AN} suggests a natural way to define ambiguity aversion for general act-to-distribution mappings. Given the act-to-distribution mapping $\mathfrak{D}$, we say that $\mathfrak{D}$ is \emph{ambiguity neutral} if there exists $\mu\in \Delta$ such that $\mathfrak{D}(X)=X_{\#}\mu$ for all $X\in \mathcal{X}$. That is, an act-to-distribution mapping is ambiguity neutral if and only if it is a probabilistic act-to-distribution mapping. Additionally, we say that an act-to-distribution mapping is ambiguity averse if it is more ambiguity averse than an ambiguity-neutral act-to-distribution mapping.

\subsection{Diversification seeking and concave act-to-distribution mappings}

\label{sec:div}

 We now extend the definition of Property (RD) to all acts. This notion was formalized in the seminal contributions of \cite{GS89} and \cite{S89},\footnote{In \cite{GS89} and \cite{S89}, this property is referred to as uncertainty aversion.} and is sometimes also described as the convexity of preferences.
\axiomplain{D}{Diversification: For all $X,Y\in \mathcal{X}$ and $\lambda\in [0,1]$, if $X\simeq Y$ then $\lambda X+(1-\lambda)Y\succsim X.$}
Clearly, Property (D) is stronger than Property (RD). However, there does not exist an equivalent proposition to Proposition \ref{prop:Extend1} in the case of Property (RD) extending to Property (D). To substantiate this claim, we provide a counterexample in Appendix \ref{app:div}. The following example motivates a property for act-to-distribution mappings that allows Property (RD) to imply Property (D).  

\begin{example}
\label{ex:cs}
    Consider a DM that perceives no ambiguity and describes their uncertainty with the unique $\mu\in \Delta$. Additionally, assume that, to aid in the decision-making, the DM uses a statistic $\gamma$ that is monotonic with respect to $\geq_{\mathrm{ssd}}$ and comonotonic quasiconcave. Therefore, the DM's preference relation $\succsim$ is given by
    $$X\succsim Y\iff \gamma\circ \mathfrak{D}_\mu(X)\geq \gamma\circ\mathfrak{D}_\mu(Y),~~X,Y\in \mathcal{X}.$$
    Establishing that $\succsim$ satisfies Property (D) follows by invoking the well-known comonotonic sum inequality; see \citet[Theorem 7]{DDGKV02}. This inequality states that, given $X,Y\in \mathcal{X}$ and comonotonic $X^c,Y^c\in \mathcal{X}$ such that $X=_{\mathrm{fsd}}^\mu X^c$ and $Y=_{\mathrm{fsd}}^\mu Y^c$, we have
    $X+Y\geq_{\mathrm{cv}}^{\mu}X^c+Y^c$. Using Proposition \ref{prop:oper}, we can rewrite this inequality in terms of the probabilistic act-to-distribution mapping as follows:
    \begin{equation}
        \label{eq:MOT1}
        \mathfrak{D}_{\mu}(X+Y)\geq_{\mathrm{cv}}\mathfrak{D}_{\mu}(X^c+Y^c)=\mathfrak{D}_{\mu}(X^c)\oplus\mathfrak{D}_{\mu}(Y^c)=\mathfrak{D}_{\mu}(X)\oplus\mathfrak{D}_{\mu}(Y).
    \end{equation} Given $\lambda\in [0,1]$, we can use \eqref{eq:MOT1} and Proposition \ref{prop:oper} to obtain the inequality
    $$\mathfrak{D}_{\mu}(\lambda X+(1-\lambda)Y)\geq_{\mathrm{cv}}\mathfrak{D}_{\mu}(\lambda X)\oplus\mathfrak{D}_{\mu}((1-\lambda)Y)=\lambda\otimes \mathfrak{D}_{\mu}(X)\oplus(1-\lambda)\otimes \mathfrak{D}_{\mu}(Y).$$
    Since $Q\geq_{\mathrm{cv}}P$ implies $Q\geq_{\mathrm{ssd}}P$ for all $Q,P\in \mathcal{D}$, we get
    \begin{equation}
        \label{eq:coIm2}
        \mathfrak{D}_{\mu}\left(\lambda X+(1-\lambda)Y\right)\geq_{\mathrm{ssd}} \lambda\otimes \mathfrak{D}_{\mu}(X)\oplus(1-\lambda)\otimes \mathfrak{D}_{\mu}(Y).
    \end{equation}
    Therefore, using \eqref{eq:coIm2} and the fact that $\gamma$ is monotonic with respect to $\geq_{\mathrm{ssd}}$ and comonotonic quasiconcave, for all $X,Y\in \mathcal{X}$ satisfying $X\simeq Y$ and $\lambda\in [0,1]$,
    \begin{equation}
        \label{eq:calc}
        \begin{aligned}
            \gamma\left(\mathfrak{D}_{\mu}\left(\lambda X+(1-\lambda)Y\right)\right)&\geq \gamma\left(\lambda\otimes \mathfrak{D}_{\mu}(X)\oplus(1-\lambda)\otimes \mathfrak{D}_{\mu}(Y)\right)\\&\geq \min\left\{\gamma(\mathfrak{D}_{\mu}(X)),\gamma(\mathfrak{D}_{\mu}(Y))\right\}=\gamma(\mathfrak{D}_{\mu}(X)).
        \end{aligned}
    \end{equation}
    Thus, $\lambda X+(1-\lambda)Y\succsim X$, which shows that $\succsim$ satisfies Property (D).
\end{example}

Therefore, motivated by \eqref{eq:coIm2} from Example \ref{ex:cs}, we introduce the following property for act-to-distribution mappings $\mathfrak{D}$, which we call \emph{concavity}:
     \begin{equation}
        \label{eq:coIm}
        \mathfrak{D}\left(\lambda X+(1-\lambda)Y\right)\geq_{\mathrm{ssd}} \lambda\otimes \mathfrak{D}(X)\oplus(1-\lambda)\otimes \mathfrak{D}(Y) .
    \end{equation}

    Let $\succsim$ be an ATE preference relation represented by \eqref{eq:AE}. If $\gamma$ is monotone with respect to $\geq_{\mathrm{ssd}}$ and comonotonic quasiconcave, and if $\mathfrak{D}$ is concave, then $\succsim$ satisfies Property (D). This follows by the same reasoning as in \eqref{eq:calc}. A natural question, which the following theorem answers, is the characterization of when Choquet act-to-distribution mappings are concave.

    \begin{theorem}
        \label{theo:SAA}
        Let $\nu$ be a continuous capacity. Then $\mathfrak{D}_{\nu}$ is concave if and only $\nu$ is supermodular.
    \end{theorem}

Concavity of an act-to-distribution mapping reflects a form of ambiguity aversion, since it implies that a strongly risk-averse DM views the distribution associated with the mixture of ambiguous acts as more attractive than the comonotonic mixture of the distributions associated with each act separately. This idea is highlighted by the fact that a Choquet act-to-distribution mapping is ambiguity averse if it is concave. This follows from Theorem \ref{theo:SAA} and the fact that supermodular capacities are balanced; see \cite{S71}. However, as the following example shows, this implication does not extend to all act-to-distribution mappings.

\begin{example}\label{ex:counter2}
    Given distinct $\mu_1,\mu_2\in \Delta$, let $\mathfrak{D}$ be the act-to-distribution mapping given by
    $$\mathrm{VaR}_{\alpha}(\mathfrak{D}(X))=(1/2)\mathrm{VaR}_{\alpha}(\mathfrak{D}_{\mu_1}(X))+(1/2)\mathrm{VaR}_{\alpha}(\mathfrak{D}_{\mu_2}(X))~\text{for all}~\alpha\in (0,1).$$
    We show in Appendix \ref{app:div} that $\mathfrak{D}$ is concave and not ambiguity averse.
\end{example}

Given a continuous capacity $\nu$, it is clear that the Choquet act-to-distribution mapping $\mathfrak{D}_{\nu}$ is concave if and only if for all $X_1,\dots,X_n\in\mathcal{X}$ and $\lambda_1,\dots,\lambda_n\in [0,1]$ satisfying $\sum_{k=1}^n \lambda_k=1$, we have
\begin{equation}
    \label{eq:newDiv}
    \mathfrak{D}_\nu\left(\sum_{k=1}^n \lambda_k X_k\right)\geq_{\mathrm{ssd}}\bigoplus_{k=1}^n \lambda_k\otimes \mathfrak{D}_\nu(X_k).
\end{equation}
By restricting the form of the acts $\sum_{k=1}^n \lambda_k X_k$ in \eqref{eq:newDiv}, we can define weaker notions of concavity for Choquet act-to-distribution mappings. Motivated by the work of \cite{HK25}, we can use this idea to characterize the family of Choquet act-to-distribution mappings presented in Example \ref{example:atd}. Given $A\in \mathcal{F}$ and $x,y\in \mathbb{R}$, $xAy$ denotes the binary act which yields $x$ if $\omega\in A$ and $y$ otherwise.

\begin{proposition}
    \label{prop:weakerDiv}
     It holds that $\mathfrak{D}_{\nu}=\mathfrak{D}_{\mathfrak{C}(\nu)}$ if and only if \eqref{eq:newDiv} holds for all   $\sum_{k=1}^n \lambda_k X_k=xAy$ where $A\in \mathcal{F}$ and $x,y\in \mathbb{R}$.
\end{proposition}

Finally, using similar arguments as the proof for Proposition \ref{prop:weakerDiv}, one can show that $\mathfrak{D}_\nu$ is ambiguity averse if and only if \eqref{eq:newDiv} holds for all $\sum_{k=1}^n \lambda_k X_k=x$ where $x\in \mathbb{R}$. This is consistent with the idea of sure diversification from \cite{CT02}.

\subsection{The distributionally robust interpretation}

\label{sec:robust}

The framework of distributionally robust optimization extends the maxmin expected utility model of \cite{GS89} by allowing for general statistics; see \cite{DKW19} and \cite{FLW24} for approaches along this line in the context of risk measures.
A preference relation $\succsim$ is called a \emph{distributionally robust} preference relation if there exists a collection of risk-conforming probability measures $\mathcal{Q}$\footnote{Here, we suppose each $\mu\in \mathcal{Q}$ is risk conforming to ensure we remain in the Savage setting with pure-risk.} and a statistic $\gamma$ such that
$$X\succsim Y\iff \min_{\mu\in \mathcal{Q}}\gamma(X_{\#}\mu)\geq\min_{\mu\in \mathcal{Q}}\gamma(Y_{\#}\mu),~~X,Y\in \mathcal{X}.$$
Intuitively, a DM with a distributionally robust preference relation views the set $\mathcal{Q}$ as a collection of equally plausible models describing their uncertainty, while the statistic $\gamma$ represents their pure-risk preferences. In this decision model, $\mathcal{Q}$ represents the ambiguity component of $\succsim$, and $\gamma$ represents the risk component of $\succsim$. Distributionally robust preference relations represent a strong form of ambiguity aversion, as discussed by \cite{GMM04}.

\begin{theorem}
    \label{theo:robust}
    Let $\succsim$ be a Choquet ATE preference relation. Then $\succsim$ is a distributionally robust preference relation if and only if the $\succsim$-matching probability is supermodular.
\end{theorem}

Since any distributionally robust preference relation is an example of an ETA preference relation, Theorem \ref{theo:robust} establishes that the ETA and ATE frameworks are capable of generating a broad class of common preference relations.

\section{Conclusion}

% 1. How do our framework and technical results improve understanding of decision under ambiguity?

% 2. What are the implications of the framework for designing new decision models?

The first conceptual contribution of this paper is to show that the role of ambiguity in a decision model can be determined by where ambiguity enters relative to the evaluation of pure risk. By formalizing the ETA and ATE frameworks in the Savage setting, we show that the classical ambiguity models can be seen as instantiations of two general thought processes. This perspective clarifies that differences between existing models arise from the structural choice of whether ambiguity is incorporated by model-specific evaluations or in the distributional representation of acts themselves. 
In turn, this structural distinction identifies the precise primitives that represent ambiguity (aggregators in the ETA case and act-to-distribution mappings in the ATE case) and shows how they interact with the statistic associated with pure risk. 

After introducing the two frameworks, we focused exclusively on the ATE framework. A key benefit of the ATE perspective is that act-to-distribution mappings offer a richer, more flexible way to incorporate ambiguity than the ETA perspective, thus  expanding the modeling domain of ambiguity attitudes by permitting them to reshape many important distributional characteristics, e.g., tail behavior, variability, and asymmetry. Such flexibility is not immediately available in the ETA framework, where it is unclear---or considerably more complicated to determine---how model-based aggregation can address these important concerns. The new ATE framework allows us to construct new decision models and opens new directions for future research. As the main ATE model analyzed in this paper, the Choquet ATE model offers great flexibility in modeling both the risk attitudes and the ambiguity attitudes. This model is only a natural first step in modeling transforms from an ambiguous act to a distribution, and there are many more possibilities for general act-to-distribution mappings to explore, drawing on tools from statistical theory. 

% These insights have direct implications for constructing new decision models under ambiguity. The ETA and ATE frameworks provide a modular blueprint: a DM may freely specify the statistic capturing pure-risk preferences and then separately choose the ambiguity component, either as an aggregator of model-specific evaluations (ETA) or as an act-to-distribution mapping (ATE). 

\label{sec:con}

\newpage

\appendix

\begin{center}
    \Large Appendices: Omitted proofs and additional results
\end{center}

Appendix \ref{app:DC} collects auxiliary results used in the proofs of the main results. Appendices \ref{app:prelim}-\ref{app:div} present the proofs omitted from the main paper and some additional results.

\section{Distorted capacities}

\label{app:DC}

This appendix introduces distorted capacities, which are a generalization of distorted probabilities, and presents several results that will be used in the proofs of subsequent appendices.

Given a distortion function $g$ and a capacity $\nu$, we call $g\circ\nu$ a \emph{distorted capacity}. If the distortion function $g$ is convex (resp.\ concave) and the capacity $\nu$ is supermodular (resp.\ submodular), then the distorted capacity $g\circ\nu$ is supermodular (resp.\ submodular); see \cite{VW25A}.

In the proofs in the following appendices, we will make use of the two following families of distortion functions. Given $\alpha\in (0,1)$, define the distortion functions $g_{\alpha}(\beta)=\id_{[1-\alpha,1]}(\beta)$ and $h_{\alpha}(\beta)=\left(\beta/\alpha-(1-\alpha)/\alpha\right)^+$ for all $\beta\in (0,1)$.
Given a capacity $\nu$ and $\alpha\in (0,1)$, we can define the following functionals given by the Choquet integral with respect to the distorted capacity: 
\begin{align*}
    \mathrm{G}_\alpha^\nu(X)=\int X \d (g_\alpha\circ \nu)~~~ \text{and}~~~\mathrm{H}_\alpha^\nu(X)=\int X \d (h_\alpha\circ \nu),~~~X\in \mathcal{X}.
\end{align*}
Since $h_\alpha$ is concave for every $\alpha\in (0,1)$, $h_\alpha\circ \nu$ is supermodular if $\nu$ is supermodular. Therefore, if $\nu$ is supermodular, then $\mathrm{H}_\alpha^\nu$ is concave for every $\alpha\in (0,1)$ by \citet[Corollary 4.2]{MM04A}. 

\begin{proposition}
    \label{prop:quant}
    Let $\nu$ be a continuous capacity. For all $\alpha\in (0,1)$, we have
    $$\mathrm{G}_\alpha^\nu(X)=\mathrm{VaR}_{\alpha}(\mathfrak{D}_{\nu}(X))~~~\text{and}~~~\mathrm{H}_\alpha^\nu(X)=\frac{1}{\alpha}\int_0^{\alpha}\mathrm{G}_\beta^\nu(X)\d\beta,~~X\in \mathcal{X}.$$
\end{proposition}
\begin{proof}
    Fix $X\in \mathcal{X}$, let $x_0=\inf\{x\in \mathbb{R}:\nu(X>x)< 1-\alpha\}$. For $\epsilon>0$, we know that 
    $\nu(X>x_0+\epsilon)< 1-\alpha$ and $\nu(X>x_0-\epsilon) \geq1-\alpha.$ Therefore, $g_\alpha(\nu(X>x_0+\epsilon))=0$ and $g_\alpha(\nu(X>x_0-\epsilon))=1$. Thus,
$$g_\alpha(\nu(X>x))=\begin{cases}
    \id_{(-\infty,x_0)}(x)&\text{if}~\nu(X>x_0)\geq 1-\alpha\\\id_{(-\infty,x_0]}(x)&\text{if}~\nu(X>x_0)< 1-\alpha
\end{cases},~~x\in \mathbb{R}.$$
By the definition of the Choquet integral, $\mathrm{G}_\alpha^\nu(X)=x_0.$ Since, by definition, $S_{\mathfrak{D}_\nu(X)}(x)=\nu(X>x)$ for all $x\in \mathbb{R}$, we have $\mathrm{G}_\alpha^\nu(X)=\mathrm{VaR}_{\alpha}(\mathfrak{D}_{\nu}(X))$. 
    
    Let $\delta\in [0,1]$, we have
    $$\frac{1}{\alpha}\int_0^{\alpha}g_{\beta}(\delta)\d\beta= \frac{1}{\alpha}\int_0^{\alpha} \id_{[1-\delta,1]}(\beta)\d\beta=\left(\delta/\alpha-(1-\alpha)/\alpha\right)^+=h_{\alpha}(\delta).$$
    Since $\varphi:\mathcal{X}\to\mathbb{R}$ defined by $\varphi(X)=\frac{1}{\alpha}\int_0^{\alpha}\mathrm{G}_\beta^\nu(X)\d\beta$ is monotone, comonotonic additive, and normalized, $\varphi(X)=\int_{\Omega}X\d\tilde{\nu}$ for some capacity $\tilde{\nu}$. For all $A\in \mathcal{F}$,
    $$\tilde{\nu}(A)=\varphi(\id_A)=\frac{1}{\alpha}\int_0^{\alpha}\mathrm{G}_\beta^\nu(\id_A)\d\beta=\frac{1}{\alpha}\int_0^{\alpha}g_{\beta}(\nu(A))\d\beta=h_{\alpha}(\nu(A)).$$
    Therefore, $\varphi(X)=\mathrm{H}_{\alpha}^\nu(X)$ for all $X\in \mathcal{X}$.
\end{proof} 

\section{Proof accompanying Section \ref{sec:prelim}}

\label{app:prelim}

\begin{proof}[Proof of Proposition \ref{prop:riskfun}]
    Let $\gamma$ and $\tilde{\gamma}$ define the same preference relation on distributions $\succsim^{\ell}$ in \eqref{eq:repLot}. Define $f:\mathbb{R}\to\mathbb{R}$ and $\tilde{f}:\mathbb{R}\to\mathbb{R}$ by
        $$f(x)=\gamma(\delta_x)~\text{and}~\tilde{f}(x)=\tilde{\gamma}(\delta_x),~~x\in \mathbb{R}.$$
        Since $\gamma$ and $\tilde{\gamma}$ are statistics, both $f$ and $\tilde{f}$ are strictly increasing and continuous. We claim that $\im(\gamma)=\im(f)$. Clearly, $\im(f)\subseteq\im(\gamma)$. Let $Q\in \mathcal{D}$. Since $Q$ has compact support, there exists $x,y\in \mathbb{R}$ such that 
$\delta_x\geq_{\mathrm{fsd}}Q\geq_{\mathrm{fsd}}\delta_y.$ Therefore, $f(x)\geq \gamma(Q)\geq f(y)$. As $f$ is continuous, by the intermediate value theorem, there exists $z\in \mathbb{R}$ such that $f(z)=\gamma(Q)$. Thus, $\im(\gamma)\subseteq\im(f)$ and $\im(\gamma)=\im(f)$. A similar argument will show that $\im(\tilde{\gamma})=\im(\tilde{f})$. Define the statistics $\psi:\mathcal{D}\to\mathbb{R}$ and $\tilde{\psi}:\mathcal{D}\to\mathbb{R}$ by 
$$\psi(Q)=f^{-1}\circ \gamma(Q)~\text{and}~\tilde{\psi}(Q)=\tilde{f}^{-1}\circ \tilde{\gamma}(Q),~~Q\in \mathcal{D},$$
where $f^{-1}:\im(f)\to \mathbb{R}$ and $\tilde{f}^{-1}:\im(\tilde{f})\to \mathbb{R}$ denote the inverses of $f$ and $\tilde{f}$, respectively. It is straightforward to show that $\psi$ and $\tilde{\psi}$ are certainty-equivalent statistics that satisfy \eqref{eq:repLot}. Note that this additionally shows the existence of a certainty-equivalent statistic representing $\succsim^\ell$. Given $Q\in \mathcal{D}$, $Q\simeq^{\ell}\delta_{\psi(Q)}$, which implies $\psi(Q)=\tilde{\psi}(Q).$ Thus, we have
$\tilde{\gamma}=\tilde{f}\circ f^{-1}\circ \gamma$.
Additionally, if $\gamma$ and $\tilde{\gamma}$ were certainty-equivalent statistics, then $\gamma=\tilde{\gamma}$ since $f(x)=\tilde{f}(x)=x$ for all $x\in \mathbb{R}$.
    \end{proof}

\section{Proofs, discussions, and results accompanying Section \ref{sec:two}}

\label{app:two}

\begin{proof}[Proof of Proposition \ref{prop:labelGood1}]
    Fix $\eta\in \mathcal{X}(\Delta,\Sigma)$.
    We claim the set $\mathcal{I}_1(\eta)=\{c\in \mathbb{R}:\eta\succsim^{m} c\}$ must be of the form $(-\infty,a_0(\eta)]$ for some $a_0(\eta)\in \mathbb{R}$. Let $a\in \mathcal{I}_1(\eta)$ and $b\in \mathbb{R}$ satisfy $a>b$, then, by Property (M3), $\eta\succsim^{m} a\succ^{m} b$ and $b\in \mathcal{I}_1(\eta)$. Find $a\in \mathbb{R}$ such that $a\geq \eta$. Therefore, by Property (M1), $a\succsim^{m} \eta$. If $b\in \mathbb{R}$ satisfies $b>a$, then, by Property (M3), $b\succ^{m} a\succsim^{m}\eta$ and $b\notin \mathcal{I}_1(\eta)$. Therefore, $\sup(\mathcal{I}_1(\eta))$ exists. Denote by $a_0(\eta)=\sup(\mathcal{I}_1(\eta))$. By Property (M2), $a_0(\eta)\in \mathcal{I}_1(\eta)$. Therefore, $\mathcal{I}_1(\eta)=(-\infty,a_0(\eta)]$. Similarly, we can show that $\mathcal{I}_2(\eta)=\{c\in \mathbb{R}:c\succsim^{m} \eta\}$ is of the form $[b_0(\eta),\infty)$ for some $b_0(\eta)\in \mathbb{R}$. It must hold that $a_0(\eta)\geq b_0(\eta)$ as $\succsim$ is a preference relation on $\mathcal{X}(\Delta,\Sigma)$. Therefore,  $\eta\simeq^{m} a_0(\eta).$
    
    Define the functional $\rho:\mathcal{X}(\Delta,\Sigma)\to\mathbb{R}$ by
    $\rho(\eta)=a_0(\eta).$ By Property (M3), it is clear that
    $$\eta\succsim^{m}\zeta\iff \rho(\eta)\geq \rho(\zeta),~~\eta,\zeta\in\mathcal{X}(\Delta,\Sigma).$$ By Property (M1), $\rho(\eta)\geq \rho(\zeta)$ for all $\eta,\zeta\in \mathcal{X}(\Delta,\Sigma)$
    satisfying $\eta\geq\zeta$. By property (M3), $\rho(a)>\rho(b)$ for all $a,b\in \mathbb{R}$ satisfying $a>b$.
\end{proof}

\begin{proposition}
    Let $\gamma$ be a statistic. Then for all $X\in \mathcal{X}$, $\gamma\circ \mathfrak{R}_X\in \mathcal{X}(\Delta,\Sigma)$.
\end{proposition}
\begin{proof}
    Fix $X\in \mathcal{X}$. Find $a,b\in \mathbb{R}$ such that $a\geq X\geq b$. Given $\mu\in \Delta$, we have $$\delta_a\geq_{\mathrm{fsd}}\mathfrak{R}_X(\mu)\geq_{\mathrm{fsd}}\delta_b.$$ Therefore, since $\gamma$ is a statistic, we have
    $\gamma(\delta_a)\geq \gamma\circ\mathfrak{R}_X(\mu)\geq \gamma(\delta_b)$. Since $\mu$ was arbitrary, $\gamma(\delta_a)\geq \gamma\circ\mathfrak{R}_X\geq \gamma(\delta_b)$. Therefore, $\gamma\circ\mathfrak{R}_X$ is bounded. To show that $\gamma\circ\mathfrak{R}_X$ is $\Sigma$-measurable, we will show that it is continuous. Let $(\mu_\lambda)_{\lambda\in \Lambda}\subseteq \Delta$ be a net converging to $\mu\in \Delta$. Given the bounded and continuous $f:\mathbb{R}\to\mathbb{R}$, we have 
    $$\lim_{\lambda}\int_{\mathbb{R}}f\d(\mathfrak{R}_X(\mu_{\lambda}))=\lim_{\lambda}\int_{\Omega}f(X)\d\mu_{\lambda}=\int_{\Omega} f(X)\d\mu =\int_{\mathbb{R}}f\d(\mathfrak{R}(\mu)).$$
    Therefore, $\mathfrak{R}_X(\mu_{\lambda})\to\mathfrak{R}_X(\mu)$ in distribution. Define $\mathcal{D}_0\subseteq \mathcal{D}$ by $$\mathcal{D}_0=\{Q\in \mathcal{D}:\supp(Q)\subseteq [\gamma(\delta_b),\gamma(\delta_a)]\}.$$
    Equipping $\mathcal{D}_0$ with the topology consistent with convergence in distribution, we have that $\gamma|_{\mathcal{D}_0}$ is continuous as $\gamma$ is a statistic. Since $\im(\mathfrak{R}_X)\subseteq \mathcal{D}_0$, we have $\lim_{\lambda}\gamma\circ\mathfrak{R}_X(\mu_\lambda)=\gamma\circ\mathfrak{R}_X(\mu)$. Therefore, $\gamma\circ\mathfrak{R}_X$ is $\Sigma$-measurable.
\end{proof}

\begin{lemma}\label{lem:FSD}
    Let $\nu$ be a continuous capacity. If $X,Y\in \mathcal{X}$ satisfy $X\geq Y$, then $\mathfrak{D}_\nu(X)\geq_{\mathrm{fsd}}\mathfrak{D}_\nu(Y)$.
\end{lemma}

\begin{proof}
    For all $x\in\mathbb{R}$, $ \{X>x\}\supseteq\{Y>x\}$. Thus, for all $x\in\mathbb{R}$,$$S_{\mathfrak{D}_{\nu}(X)}(x)=\nu(X>x)\geq \nu(Y>x)=S_{\mathfrak{D}_{\nu}(Y)}(x).$$ Therefore, $\mathfrak{D}_{\nu}(X)\geq_{\mathrm{fsd}}\mathfrak{D}_{\nu}(Y)$.
\end{proof}

\begin{lemma}\label{lem:con}
    Let $\nu$ be a continuous capacity. If $(Z_n)_{n\in \N}\subseteq \mathcal{X}$ is a bounded sequence pointwise converging to $Z\in \mathcal{X}$, then $\mathfrak{D}_\nu(Z_n)\to \mathfrak{D}_{\nu}(Z)$ in distribution.
\end{lemma}

\begin{proof}
      For $n\in\N$, define $X_n=\sup_{m\geq n}Z_m$. It holds that $(X_n)_{n\in \N}$ is a bounded sequence such that $X_n\downarrow Z$ pointwise. We claim that $\mathfrak{D}_{\nu}(X_n)\to \mathfrak{D}_{\nu}(Z)$ in distribution.
      Denote by $\mathfrak{Z}\subseteq \mathbb{R}$ the set of continuity points for $S_{\mathfrak{D}_{\nu}(Z)}$. Given $x\in \mathfrak{X}$, since $$\{Z>x\}\subseteq \bigcap_{n=1}^\infty\{X_n>x\}\subseteq \{Z\geq x\},$$ it holds that $\nu(Z>x)=\nu(Z\geq x)\geq\nu(\bigcap_{n=1}^\infty\{X_n>x\})\geq \nu(Z> x).$ Additionally, as $\nu$ is continuous, we have
    $$S_{\mathfrak{D}_\nu(X)}(x)=\nu(X>x)=\nu\left(\bigcap_{n=1}^\infty\{X_n>x\}\right)=\lim_{n\to\infty}\nu(X_n>x)=\lim_{n\to\infty}S_{\mathfrak{D}_\nu(X_n)}(x)$$
    Therefore $\mathfrak{D}_{\nu}(X_n)\to\mathfrak{D}_{\nu}(Z)$ in distribution. 

    For $n\in \N$, define $Y_n=\inf_{m\geq n}Z_m$. It holds that $(Y_n)_{n\in \N}$ is a bounded sequence such that $Y_n\uparrow Z$ pointwise.  Since $\nu$ is continuous,
$$\lim_{n\to\infty}S_{\mathfrak{D}_\nu(Y_n)}(x)=\lim_{n\to\infty}\nu(Y_n>x)=\nu\left(\bigcup_{n=1}^\infty \{Y_n>x\}\right)=\nu(Z>x)=S_{\mathfrak{D}_\nu(Z)}(x),~~x\in \mathbb{R}.$$
    Thus, $\mathfrak{D}_\nu(Y_n)\to \mathfrak{D}_\nu(Z)$ in distribution.

    Finally, since $X_n\geq Z_n\geq Y_n$ for all $n\in \N$, by Lemma \ref{lem:FSD}, we have that $$\mathfrak{D}_\nu(X_n)\geq_{\mathrm{fsd}}\mathfrak{D}_\nu(Z_n)\geq_{\mathrm{fsd}}\mathfrak{D}_\nu(Y_n).$$
    Therefore, if $x\in \mathbb{R}$ is a continuity point of $S_{\mathfrak{D}_\nu(Z)}$, by the above results, we have
    $$
    \begin{aligned}
      S_{\mathfrak{D}_\nu(Z)}(x)=\lim_{n\to\infty}S_{\mathfrak{D}_\nu(X_n)}(x)&\geq \limsup_{n\to\infty}S_{\mathfrak{D}_\nu(Z_n)}(x)\\&\geq \liminf_{n\to\infty}S_{\mathfrak{D}_\nu(Z_n)}(x)\geq \lim_{n\to\infty}S_{\mathfrak{D}_\nu(Y_n)}(x)=S_{\mathfrak{D}_\nu(Z)}(x).  
    \end{aligned}
    $$
    Therefore, $\mathfrak{D}_\nu(Z_n)\to\mathfrak{D}_\nu(Z)$ in distribution.
\end{proof}

\begin{proof}[Proof of Theorem \ref{theo:riskEq}]
    Assume that $\mathfrak{D}$ satisfies Properties (R1)--(R3). Let $f:\mathbb{R}\to\mathbb{R}$ be continuous and increasing. We will first show that \begin{equation}
        \label{eq:incr}
        \mathfrak{D}(f(X))=f_{\#}(\mathfrak{D}(X)) \footnote{Note that Property (R3) only guarantees this equality for continuous and strictly increasing functions.},~~X\in \mathcal{X}.
    \end{equation} For each $n\in \N$, define the function $f_n:\mathbb{R}\to\mathbb{R}$ by
    $f_n(x)=f(x)+\exp(x-n)$. The following are immediately clear. (1) For all $x\in \mathbb{R}$, $\lim_{n\to\infty}f_n(x)=f(x)$. (2) For all $n\in \N$, $f_n$ is a vNM utility function. (3) For all $X\in \mathcal{X}$, $(f_n(X))_{n\in \N}$ is bounded. Let $Q\in \mathcal{D}$, we claim that $(f_n)_{\#}Q\to f_\# Q$ in distribution. To see this, let $g:\mathbb{R}\to\mathbb{R}$ be bounded and continuous. Then, by the dominated convergence theorem, we have
    $$\int_{\mathbb{R}}g\d((f_n)_{\#}Q)=\int_{\mathbb{R}}g\circ f_n\d Q\to \int_{\mathbb{R}}g\circ f\d Q=\int_{\mathbb{R}}g\d(f_{\#}Q).$$
    Therefore, $(f_n)_{\#}Q\to f_\# Q$ in distribution. By Property (R2) and (R3), for all $X\in \mathcal{X}$, we have
    $$\mathfrak{D}(f(X))=\lim_{n\to\infty}\mathfrak{D}(f_n(X))=\lim_{n\to\infty}(f_n)_{\#}(\mathfrak{D}(X))=f_{\#}(\mathfrak{D}(X)).$$

    Next, we claim that \begin{equation}
        \label{eq:supp}\supp(\mathfrak{D}(X))\subseteq \overline{\im(X)}, ~~X\in\mathcal{X},
    \end{equation} where $\overline{\im(X)}$ denotes the closure of $\im(X)$. To see this claim, for the sake of contradiction, assume that \eqref{eq:supp} is false. That is, assume there exists $X\in \mathcal{X}$ and $x_0\in \supp(\mathfrak{D}(X))$ such that $x_0\notin \overline{\im(X)}$. To ease notation, let $Q_X=\mathfrak{D}(X)$. Therefore, we can find $\epsilon>0$ such that for $I= (x_0-\epsilon,x_0+\epsilon)$, we have $Q_X(I)>0$ and $X^{-1}(I)=\varnothing.$ Define $I_1=(x_0-\epsilon,x_0]$ and $I_2=(x_0,x_0+\epsilon)$. Without loss of generality, we may assume  $Q_X(I_1)>0$. Define the function $f:\mathbb{R}\to\mathbb{R}$ by 
    $$f(x)=\begin{cases} x_0-\epsilon&\text{if}~x\in I_1\\ 2x-x_0-\epsilon &\text{if}~x\in I_2\\ x &\text{if}~x\notin I_1\cup I_2\end{cases}.$$
    Let $g$ denote the identity function. We have that $f(X)=g(X)$. Therefore, by \eqref{eq:incr}, $f_{\#}Q_X=g_{\#}Q_X.$
    However, we have
    $$\begin{aligned}
       \relax[f_{\#}Q_X](\{x_0-\epsilon\})&=Q_X([x_0-\epsilon,x_0])\\&=Q_X(\{x_0-\epsilon\})+Q_X(I_1)>Q_X(\{x_0-\epsilon\})=[g_{\#}Q_X](\{x_0-\epsilon\}), 
    \end{aligned}
    $$
    a contradiction. 
    
    Let $X,Y\in \mathcal{X}$ be comonotonic, by \citet[Proposition 4.5]{D94A}, there exist  increasing and continuous $f:\mathbb{R}\to\mathbb{R}$ and $g:\mathbb{R}\to\mathbb{R}$ such that $f(X+Y)=X$, $g(X+Y)=Y$, and $f(x)+g(x)=x$ for all $x\in \mathbb{R}$. Therefore, by \eqref{eq:incr} and \citet[Theorem 1]{LMW24A}, for all $\alpha\in (0,1)$,
    $$\begin{aligned}
        \VaR_{\alpha}(\mathfrak{D}(X)\oplus \mathfrak{D}(Y))&=\VaR_{\alpha}(\mathfrak{D}(X))+\VaR_{\alpha}(\mathfrak{D}(Y))\\&=\VaR_{\alpha}(f_{\#}(\mathfrak{D}(X+Y)))+\VaR_{\alpha}(g_{\#}(\mathfrak{D}(X+Y)))\\&=f(\VaR_{\alpha}(\mathfrak{D}(X+Y)))+g(\VaR_{\alpha}(\mathfrak{D}(X+Y)))\\&=\VaR_{\alpha}(\mathfrak{D}(X+Y)).
    \end{aligned}$$
    Thus $\mathfrak{D}(X+Y)=\mathfrak{D}(X)\oplus \mathfrak{D}(Y)$. For $X\in \mathcal{X}$ and positive $a\in \mathbb{R}$, by Property (R3), we have $\mathfrak{D}(aX)=a\otimes \mathfrak{D}(X)$. If $a=0$, we have $\mathfrak{D}(aX)=a\otimes \mathfrak{D}(X)$ by \eqref{eq:supp}.

    Given $A\in \mathcal{F}$, by \eqref{eq:supp}, we know that there exists $a_A\in [0,1]$ such that $$\mathfrak{D}(\id_A)=a_A\delta_1+(1-a_A)\delta_0.$$
    Define the function $\nu:\mathcal{F}\to[0,1]$ by $\nu(A)=a_A$. By \eqref{eq:supp}, $\nu(\varnothing)=0$ and $\nu(\Omega)=1$. By Property (R1), for all $A,B\in \mathcal{F}$ satisfying $A\subseteq B$, $\nu(A)\leq \nu(B)$. Therefore, $\nu$ is a capacity. Let $(A_n)_{n\in\N}\subseteq \mathcal{F}$ be an increasing sequence. Define $A=\bigcup_{n=1}^\infty A_n$. By Property (R2), we have $\mathfrak{D}(\id_{A_n})\to\mathfrak{D}(\id_A)$ in distribution. Therefore, since $\mathbb{E}$ is a statistic,
$$\lim_{n\to\infty}\nu(A_n)=\lim_{n\to\infty}\mathbb{E}(\mathfrak{D}(\id_{A_n}))= \mathbb{E}(\mathfrak{D}(\id_A))=\nu(A).$$ Thus, $\nu$ is upwards continuous. A similar proof will show that $\nu$ is downwards continuous. Thus, $\nu$ is continuous. Remark that for all $A\in \mathcal{F}$, $\mathfrak{D}(\id_A)=\mathfrak{D}_\nu(\id_A)$.

We claim that for all non-negative simple $X\in \mathcal{X}$, $\mathfrak{D}(X)=\mathfrak{D}_{\nu}(X)$. Fix a non-negative simple $X\in \mathcal{X}$. We can find $N\in \N$, $(a_n)_{n=1}^N\subseteq \mathbb{R}$ satisfying $0<a_1<\dots<a_N$, and $(A_n)_{n=1}^N\subseteq \mathcal{F}$ satisfying $\Omega=A_1\supset A_2\supset \dots \supset A_N\neq \varnothing$ such that $X=\sum_{n=1}^N(a_n-a_{n-1})\id_{A_n}$, where $a_0=0$. For all $\alpha\in (0,1)$, $$\begin{aligned}
    \VaR_{\alpha}(\mathfrak{D}(X))&=\VaR_{\alpha}\left(\bigoplus_{n=1}^N (a_n-a_{n-1})\otimes \mathfrak{D}(\id_{A_n})\right)\\&=\sum_{n=1}^N(a_n-a_{n-1})\VaR_\alpha(\mathfrak{D}(\id_{A_n}))\\&=\sum_{n=1}^N(a_n-a_{n-1})\VaR_\alpha(\mathfrak{D}_{\nu}(\id_{A_n}))=\VaR_{\alpha}(\mathfrak{D}_{\nu}(X)).
\end{aligned}$$
Therefore, $\mathfrak{D}(X)=\mathfrak{D}_{\nu}(X)$.

We claim that for all non-negative $X\in \mathcal{X}$, $\mathfrak{D}(X)=\mathfrak{D}_{\nu}(X)$. Fix a non-negative $X\in \mathcal{X}$. Find a sequence $(X_n)_{n\in \N}\subseteq \mathcal{X}$ of non-negative simple acts such that $X_{n+1}\geq X_n\geq 0$ for all $n\in \N$ and $X_n\to X$ pointwise. By Property (R2) and Lemma \ref{lem:con}, we have
$$\mathfrak{D}(X)=\lim_{n\to\infty}\mathfrak{D}(X_n)=\lim_{n\to\infty}\mathfrak{D}_{\nu}(X_n)=\mathfrak{D}_{\nu}(X).$$
Therefore, $\mathfrak{D}(X)=\mathfrak{D}_{\nu}(X)$.

Finally, we claim that for general $X\in \mathcal{X}$, $\mathfrak{D}(X)=\mathfrak{D}_{\nu}(X)$. Fix $X\in \mathcal{X}$. We can find $a\in \mathbb{R}$ such that $X+a\geq 0$. Since $X+a$ and $-a$ are comonotonic, we have 
$$\mathfrak{D}(X)=\mathfrak{D}(X+a-a)=\mathfrak{D}(X+a)\oplus (-a)=\mathfrak{D}_\nu(X+a)\oplus (-a)=\mathfrak{D}_\nu(X).$$
Therefore, $\mathfrak{D}(X)=\mathfrak{D}_{\nu}(X)$.

Conversely, let $\nu$ be a continuous capacity. The fact that $\mathfrak{D}_\nu$ satisfies Properties (R1) and (R2) follows from Lemma \ref{lem:FSD} and Lemma \ref{lem:con}, respectively. Let $u$ be a vNM utility function. By \citet[Theorem 1]{LMW24A} and Proposition \ref{prop:quant}, for all $\alpha\in (0,1)$,
$$\VaR_\alpha(\mathfrak{D}_\nu(u(X)))=G_{\alpha}^\nu(u(X))=u(G_{\alpha}^\nu(X))=u(\VaR_\alpha(\mathfrak{D}_\nu(X))=\VaR_\alpha(u_{\#}(\mathfrak{D}_\nu(X))).$$
Therefore, $\mathfrak{D}_\nu(u(X))=u_{\#}(\mathfrak{D}_\nu(X)).$ Thus $\mathfrak{D}_\nu$ satisfies Property (R3).
\end{proof}

\section{Proofs accompanying
Section \ref{sec:Axioms}}

\begin{proof}[Proof of Theorem \ref{theo:mainrep}]
    Assume that $\succsim$ satisfies Axioms  (M), (RC), (SRM), (C), and (CD). We claim that for all $X\in \mathcal{X}$, there exists a unique $c_X\in \mathbb{R}$ such that $X\simeq c_X$. Let $X\in \mathcal{X}$, find $a,b\in \mathbb{R}$ such $a\geq X\geq b$. Therefore, by Axiom (M), the sets $\mathcal{L}_X=\{c\in \mathbb{R}:X\succsim c \}$ and $\mathcal{U}_X=\{c\in \mathbb{R}:c\succsim X\}$ are non-empty. It must hold that $\mathcal{L}_X\subseteq (-\infty,a]$, because if $X\succsim c$ for $c\in \mathbb{R}$ with $c>a$, then by Axioms (MO) and (SRM), it would follow that
    $X\succ a\succsim X$. Similarly, we have that $\mathcal{U}_X\subseteq [b,\infty).$ Let $a^*=\sup(\mathcal{L}_X)\in \mathbb{R}$ and $b^*=\inf(\mathcal{U}_X)\in \mathbb{R}$. By Axiom (C), it is clear that
    $\mathcal{L}_X=(-\infty,a^*]$ and $\mathcal{U}_X=[b^*, \infty).$ As $\mathcal{L}_X\cup \mathcal{U}_X=\mathbb{R}$, $b^*\leq a^*$. Take $c_X\in \mathcal{L}_X\cap \mathcal{U}_X$. Uniqueness of $c_X$ follows from Axiom (SRM). Therefore, by Axiom (SRM),
    $$X\succsim Y\iff c_X\geq c_Y,~~X,Y\in \mathcal{X}.$$
    
    Let $\nu$ denote the $\succsim$-matching probability. By Proposition \ref{lemma:match}, $\nu$ is continuous and risk conforming. As $(\Omega,\mathcal{G},\mathbb{P})$ is atomless, given any $Q\in \mathcal{D}$, there exists $X\in \mathcal{X}(\mathcal{G})$ such that $\mathfrak{D}_{\nu}(X)=Q$. Define $\gamma:\mathcal{D}\to \mathbb{R}$ by
    $\gamma(Q)=c_X,$ where $\mathfrak{D}_{\nu}(X)=Q$. To show that $\gamma$ is well-defined, we will show that
    $$\mathfrak{D}_{\nu}(X)=\mathfrak{D}_{\nu}(Y)\implies c_X= c_Y,~~X,Y\in \mathcal{X}.$$ Let $X,Y\in \mathcal{X}$ satisfy $\mathfrak{D}_{\nu}(X)=\mathfrak{D}_{\nu}(Y)$. Therefore, for all $x\in \mathbb{R}$, $\nu(X>x)=\nu(Y>x)$. By Proposition \ref{lemma:match}, $\{X>x\}\simeq \{Y>x\}.$ As $x\in \mathbb{R}$ was general, by Axiom (CD), $X\simeq Y$ and $c_X=c_Y$. As $c_X=\gamma\circ \mathfrak{D}_{\nu}(X)$ for all $X\in \mathcal{X}$,
    $$X\succsim Y\iff c_X\geq c_Y\iff \gamma\circ \mathfrak{D}_{\nu}(X)\geq \gamma\circ \mathfrak{D}_{\nu}(Y),~~X,Y\in \mathcal{X}.$$
    
    We need to show that $\gamma$ is a statistic. Let $Q,P\in \mathcal{D}$ satisfy $Q\geq_{\mathrm{fsd}} P.$ As $(\Omega,\mathcal{G},\mathbb{P})$ is atomless, there exists a $\mathcal{G}$-measurable $U:\Omega\to(0,1)$ with $\mathbb{P}(U\leq x)=x$ for all $x\in (0,1)$, i.e., $U$ has a uniform distribution under $\mathbb{P}$. It is routine to check that $\mathcal{G}$-measurable acts $X_Q=q_Q(U)$ and $X_P=q_P(U)$ satisfy $\mathfrak{D}_{\nu}(X_Q)=Q$ and $\mathfrak{D}_{\nu}(X_P)=P$. As $X_Q\geq X_P$, by Axiom (M),
    $$\gamma(Q)=\gamma\circ \mathfrak{D}_{\nu}(X_Q)\geq \gamma\circ \mathfrak{D}_{\nu}(X_P)=\gamma(P).$$
   
    Let $Q,P\in \mathcal{D}$ satisfy $Q>_{\mathrm{fsd}} P.$ Define $X_Q,X_P\in \mathcal{X}(\mathcal{G})$ the same way as above. As $X_Q\geq X_P$, we have that $X_Q\geq^{\mathbb{P}}_{\mathrm{as}}X_P$ We will show that $X_Q>^{\mathbb{P}}_{\mathrm{as}}X_P$. As $Q>_{\mathrm{fsd}} P$, there exists $\alpha_0\in (0,1)$, such that $q_Q(\alpha_0)>q_P(\alpha_0)$. By the right continuity of $q_Q$ and $q_P$, there exists $\epsilon\in(0,1-\alpha_0)$ such that $q_Q(\alpha)>q_P(\alpha)$ for all $\alpha\in [\alpha_0,\alpha_0+\epsilon]$. Therefore, $\{\alpha_0<U\leq \alpha_0+\epsilon\}\subseteq\{X_Q>X_P\}.$ Since $\mathbb{P}(\alpha_0<U\leq \alpha_0+\epsilon)=\epsilon>0$, $X_Q>^{\mathbb{P}}_{\mathrm{as}}X_P$. By Axiom (SRM),
    $$\gamma(Q)=\gamma\circ \mathfrak{D}_{\nu}(X_Q)> \gamma\circ \mathfrak{D}_{\nu}(X_P)=\gamma(P).$$
    
    Finally, let $(Q_n)_{n\in \N}\subseteq \mathcal{D}$ be a sequence with uniformly bounded support converging in distribution to $Q\in \mathcal{D}$. Let $\mathfrak{X}\subseteq (0,1)$ denote the points of discontinuity for $q_{Q}$. Since $q_{Q}$ is increasing, $\mathfrak{X}$ is at most countable. Therefore, $\mathfrak{X}$ is Borel measurable. Define $A=U^{-1}(\mathfrak{X})\in \mathcal{G}$. As $U$ has a uniform distribution under $\mathbb{P}$, $\mathbb{P}(A)=\mathbb{P}(U\in \mathfrak{X})=0$. For all $n\in \N$, define the $\mathcal{G}$-masurable acts $X_n=\id_{A^c}q_{Q_n}(U)$ and $X=\id_{A^c} q_{Q}(U)$. Since $\mathbb{P}(A^c)=1$, for all $n\in \N$, $X_n=_{\mathbb{P}}^{\mathrm{as}}q_{Q_n}(U)$ and $X=_{\mathbb{P}}^{\mathrm{as}} q_{Q}(U)$. Therefore, $\mathfrak{D}_{\nu}(X_n)=Q_n$ and $\mathfrak{D}_{\nu}(X)=Q$. We claim that $X_n\to X$ pointwise. Let $\omega\in A^c$, then, by the definition of $A$, $U(\omega)$ is a point of continuity of $q_Q$. Therefore,
    $$\lim_{n\to\infty}X_n(\omega)=\lim_{n\to\infty}q_{Q_n}(U(\omega))= q_{Q}(U(\omega))=X(\omega),$$
    as $Q_n\to Q$ in distribution. If $\omega\in A$, for all $n\in \N$, $X_n(\omega)=X(\omega)=0$. Therefore, $X_n\to X$ pointwise. Note that $(X_n)_{n\in \N}\subseteq \mathcal{X}(\mathcal{G})$ is a bounded sequence.  For each $n\in N$, define $Z_n=\inf_{m\geq n}X_m$. The sequence $(Z_n)_{n\in \N}$ is bounded and $Z_n\to X$ pointwise. 
    As $X_m\geq Z_n$ for all $n\in \N$ and $m\geq n$, by Axiom (M), $X_m\succsim Z_n$ for all $n\in \N$ and $m\geq n$. Therefore,
    $$\gamma(Q_m)=\gamma\circ\mathfrak{D}_{\nu}(X_m)\geq \gamma\circ\mathfrak{D}_{\nu}(Z_n),~~n\in \N~\text{and}~m\geq n.$$
    Thus, for all $n\in \N$, $\inf_{m\geq n}\gamma(Q_m)\geq \gamma\circ\mathfrak{D}_{\nu}(Z_n).$ 
    Therefore,
    $\liminf_{m\to\infty} \gamma(Q_m)\geq \gamma\circ \mathfrak{D}_{\nu}(Z_n)$ for all $n\in \N$, implying that $\liminf_{m\to\infty} \gamma(Q_m)\succsim Z_n$ for all $n\in \N$. By Axiom (C), $\liminf_{m\to\infty} \gamma(Q_m)\succsim X$. Therefore,
    $$\liminf_{m\to\infty} \gamma(Q_m)\geq \gamma
    \circ\mathfrak{D}_{\nu}(X)=\gamma(Q).$$
    One can show that $\gamma(Q)\geq \limsup_{m\to\infty} \gamma(Q_m)$ in a similar fashion. Therefore, $\lim_{n\to\infty}\gamma(Q_n)=\gamma(Q)$. Thus $\gamma$ is a statistic.

    For the converse, assume that there exists a statistic $\gamma$ and a risk-conforming capacity $\nu$ 
     such that \eqref{eq:mainrep} holds. As $\gamma$ is a statistic, by Lemma \ref{lem:FSD}, $\succsim$ satisfies Axiom (M). Since $\nu$ is risk conforming, it is clear that $\succsim$ satisfies Axiom (RC). Since  $\nu$ is risk conforming and $\gamma$ is a statistic, it is clear that $\succsim$ satisfies Axiom (SRM). Since $\gamma$ is a statistic, by Lemma \ref{lem:con}, $\succsim$ satisfies Axiom (C).    

    Define the function $\varphi:[0,1]\to[0,1]$ by $\varphi(x)=\gamma(x\delta_1+(1-x)\delta_0)$. As $\gamma$ is a statistic, $\varphi$ is distortion function. Let $X,Y\in \mathcal{X}$ satisfy $X\simeq^* Y$. Therefore, given $x\in \mathbb{R}$,
    $$\varphi(\nu(X>x))=\gamma\circ \mathfrak{D}_{\nu}\left(\id_{\{X>x\}}\right)=\gamma\circ\mathfrak{D}_{\nu}\left(\id_{\{Y>x\}}\right)=\varphi(\nu(Y>x)).$$ Therefore, $\nu(X>x)=\nu(Y>x)$. As $x\in \mathbb{R}$ was general, $\mathfrak{D}_{\nu}(X)=\mathfrak{D}_{\nu}(Y)$. Therefore, $\gamma\circ \mathfrak{D}_{\nu}(X)=\gamma\circ \mathfrak{D}_{\nu}(Y)$ and $X\simeq Y$. Thus $\succsim$ satisfies Axiom (CD).

    The proof of the uniqueness claim for $\gamma$ in \eqref{eq:mainrep} follows from Proposition \ref{prop:riskfun}. Finally, we will show that $\nu$ in \eqref{eq:mainrep} is the $\succsim$-matching probability. Let $A\in \mathcal{F}$ and $R_A\in \mathcal{G}$ with $\mathbb{P}(R_A)=\nu(A)$. Then, using the function $\varphi$ defined above, 
    $$\gamma\circ \mathfrak{D}_{\nu}(\id_A)=\varphi(\nu(A))=\varphi(\mathbb{P}(R_A))=\varphi(\nu(R_A))=\gamma\circ \mathfrak{D}_{\nu}(\id_{R_A}).$$
    Therefore $A\simeq R_A$ and $\nu$ is the $\succsim$-matching probability.
\end{proof}

\section{Proofs accompanying
Section \ref{sec:RA}}

\label{app:div}

\begin{proof}[Proof of Proposition \ref{prop:oper}]
    Let $X,Y\in \mathcal{X}$ be comonotonic. Then, by Proposition \ref{prop:quant}, for all $\alpha\in (0,1)$ it holds that 
    \begin{align*}
    \mathrm{VaR}_\alpha(\mathfrak{D}_\nu(X+Y))=\mathrm{G}_\alpha^\nu(X+Y)=\mathrm{G}_\alpha^\nu(X)+\mathrm{G}_\alpha^\nu(Y)&=\mathrm{VaR}_\alpha(\mathfrak{D}_\nu(X))+\mathrm{VaR}_\alpha(\mathfrak{D}_\nu(Y))\\&=\mathrm{VaR}_\alpha(\mathfrak{D}_\nu(X)\oplus\mathfrak{D}_\nu(Y)),
    \end{align*}
    as $X$ and $Y$ are comonotonic and $\mathrm{G}_\alpha^\nu$ is comonotonic additive. Given $\lambda\in [0,1]$, an identical proof will show that $\mathfrak{D}_\nu(\lambda X)=\lambda \otimes \mathfrak{D}_\nu(X)$ since $\mathrm{G}_\alpha^\nu(aX)=a\mathrm{G}_\alpha^\nu(X)$ for all $X\in \mathcal{X}$ and non-negative $a\in \mathbb{R}$.
\end{proof}

\begin{proof}[Proof of Proposition \ref{prop:coDi}]
    Let $\gamma$ denote the $\succsim$-CES and $\nu$ denote the $\succsim$-matching probability. Assume that $\succsim$ satisfies Property (WRD). Fix $Q,P\in \mathcal{D}$, without loss of generality, we may assume that $\gamma(Q)\geq \gamma(P)$. Given $\lambda\in [0,1]$, we need to show that
    $\gamma(\lambda\otimes Q\oplus (1-\lambda)\otimes P)\geq \gamma(P)$. Since $(\Omega,\mathcal{G},\mathbb{P}|_{\mathcal{G}})$ is atomless, there exists a $\mathcal{G}$-measurable $U:\Omega\to(0,1)$ with $\mathbb{P}(U\leq x)=x$ for all $x\in (0,1)$, i.e., $U$ has a uniform distribution under $\mathbb{P}$. Define $X_Q=q_Q(U)$ and $X_P=q_P(U)$, which are both elements of $\mathcal{X}(\mathcal{G})$. Since the $\succsim$-matching probability is risk conforming, this implies that $\mathfrak{D}_\nu(X_Q)=Q$ and $\mathfrak{D}_\nu(X_P)=P$. Since $\gamma(Q)\geq \gamma(P)$, this implies that $X_Q\succsim X_P$. 
    
    We claim that there exists $c_0\in [0,\infty)$ such that $X_Q-c_0\simeq X_P.$ Define $$c_0=\sup\{c\in [0,\infty): X_Q-c\succsim X_P\}.$$
    Let $c_1=\|X_Q\|+\|X_P\|$. We have that 
    $X_Q-c_1\leq -\|X_P\|\leq X_P.$ By Axioms (MO) and (SRM), for all $c>c_1$, $X_P\succ X_Q-c$. Thus $\{c\in [0,\infty): X_Q-c\succsim X_P\}\subseteq [0,c_1]$ and $c_0\in [0,\infty)$. A simple consequence of Axiom (C) is that $X_Q-c_0\simeq X_P$.  

    It is clear that $X_Q$ and $X_P$ are comonotonic. Therefore, $X_Q-c_0$ and $X_P$ are comonotonic and, by Property (WRD),
    $\lambda(X_Q-c_0)+(1-\lambda)X_P\succsim X_P$. Thus, $$\gamma\circ\mathfrak{D}_\nu\left(\lambda (X_Q-c_0)+(1-\lambda)X_P\right)\geq \gamma\circ\mathfrak{D}_\nu(X_P)=\gamma(P).$$
    Also, by Proposition \ref{prop:oper} and Axiom (MO), we have
    $$\gamma(\lambda\otimes Q\oplus (1-\lambda)\otimes P)=\gamma\circ\mathfrak{D}_\nu(\lambda X_Q+(1-\lambda)X_P)\geq \gamma\circ\mathfrak{D}_\nu(\lambda (X_Q-c_0)+(1-\lambda)X_P).$$
    Thus, it holds that $\gamma(\lambda\otimes Q\oplus (1-\lambda)\otimes P)\geq \gamma(P)$.

    Conversely, assume that $\gamma$ is comonotonic quasiconcave. Let $X,Y\in \mathcal{X}(\mathcal{G})$ be comonotonic such that $X\simeq Y$ and $\lambda\in [0,1]$. As $\lambda X$ and $(1-\lambda)\lambda Y$ are comonotonic, by Proposition \ref{prop:oper},
    $$\mathfrak{D}_\nu(\lambda X+(1-\lambda)Y)=\lambda\otimes \mathfrak{D}_\nu(X)\oplus(1-\lambda)\otimes \mathfrak{D}_\nu(Y).$$
    Since $\gamma$ is comonotonic quasiconcave,
    $$\begin{aligned}
        \gamma\circ\mathfrak{D}_\nu\left(\lambda X+(1-\lambda )Y\right)&=\gamma\left(\lambda\otimes \mathfrak{D}_\nu(X)\oplus(1-\lambda)\otimes \mathfrak{D}_\nu(Y)\right)\\&\geq \min\left\{\gamma\circ\mathfrak{D}_\nu(X), \gamma\circ\mathfrak{D}_\nu(Y)\right\}=\gamma\circ\mathfrak{D}_\nu(X).
    \end{aligned}$$
    Therefore, $\lambda X+(1-\lambda)Y\succsim X$.
\end{proof}

\begin{proof}[Proof of Theorem \ref{theo:rD}]
    The forward is trivial. The converse follows from Example \ref{ex:cs}. 
\end{proof}

\begin{proof}[Proof of Proposition \ref{prop:Extend1}]
    This is a direct consequence of Proposition \ref{prop:oper} and Proposition \ref{prop:coDi}.
\end{proof}

\begin{proof}[Proof of Proposition \ref{prop:riskEx}]
        Results (i) and (ii) follow directly from Proposition \ref{prop:oper}. To prove (iii), let $\gamma$ denote the $\succsim$-CES and $\nu$ denote the $\succsim$-matching probability. Since $\gamma$ is biseparable, there exists a vNM utility $u$ and a statistic $\tilde{\gamma}$ satisfying constant additivity and positive homogeneity such that \eqref{eq:bis} holds. Define $I:\mathcal{X}\to\mathbb{R}$ by $I(X)=\tilde{\gamma}
        \circ \mathfrak{D}_\nu(X)$. It is clear that  $I(X)\geq I(Y)$ for all $X,Y\in\mathcal{X}$ with $X\geq Y$. By Proposition \ref{prop:oper}, $I(aX+b)=aI(X)+b$ for all $X\in \mathcal{X}$, non-negative $a\in \mathbb{R}$, and $b\in \mathbb{R}$. By Theorem \ref{theo:riskEq}, we have
        $$I(u(X))=\tilde{\gamma}\circ\mathfrak{D}_{\nu}(u(X))=\tilde{\gamma}(u_{\#}(\mathfrak{D}_\nu(X)))=\gamma\circ\mathfrak{D}_\nu(X),~~X\in \mathcal{X}.$$
        Therefore, 
        $$X\succsim Y\iff \gamma\circ\mathfrak{D}_\nu(X)\geq \gamma\circ\mathfrak{D}_\nu(Y)\iff I(u(X))\geq I(u(Y)),~~X,Y\in \mathcal{X},$$
        and the proof is complete.
    \end{proof}

\begin{proof}[Proof of Theorem \ref{prop:comAm}]
    Assume that $\succsim_2$ is more ambiguity averse than $\succsim_1$. Let $A\in \mathcal{F}$ and $R_1\in \mathcal{G}$ such that $A\simeq_1 R_1$. By the definition of comparative ambiguity attitudes, $R_1\succsim_2 A$. Find $R_2\in \mathcal{G}$ such that $A\simeq_2 R_2$. As $R_1 \succsim_2 R_2$, it holds that $\mathbb{P}(R_1)\geq \mathbb{P}(R_2)$ and $\nu_1(A)\geq \nu_2(A).$ To show the converse, let $R\in \mathcal{G}$ and $A\in \mathcal{F}$ satisfy $R\succsim_1 A$. Furthermore, let $R_1,R_2\in \mathcal{G}$ satisfy $A\simeq_1 R_1$ and $A\simeq_2 R_2$. Since $\nu_1(A)\geq \nu_2(A)$ and $R\succsim_1 A$, $\mathbb{P}(R)\geq\mathbb{P}(R_1)\geq \mathbb{P}(R_2)$. Therefore, $R\succsim_2 R_2\simeq_2 A$.
\end{proof}

\begin{proof}[Proof of Proposition \ref{prop:attitudes}]
    Assume that the $\succsim_1$-CES is equal to the $\succsim_2$-CES and $\succsim_2$ is more ambiguity averse than $\succsim_1$. Denote the shared certainty-equivalent statistic by $\gamma$. Denote the $\succsim_1$-matching probability by $\nu_1$ and the $\succsim_2$-matching probability by $\nu_2$. By Theorem \ref{prop:comAm}, for all $A\in \mathcal{F}$, $\nu_1(A)\geq \nu_2(A)$. Therefore, for all $X\in \mathcal{X}$,
    $\nu_1(X>x)\geq \nu_2(X>x)$ for all $x\in \mathbb{R}.$
    Thus, for all $X\in \mathcal{X}$, $\mathfrak{D}_{\nu_1}(X)\geq_{\mathrm{fsd}} \mathfrak{D}_{\nu_2}(X)$. Since $\gamma$ is a statistic, for all $X\in \mathcal{X}$, we have $\gamma\circ\mathfrak{D}_{\nu_1}(X)\geq \gamma\circ\mathfrak{D}_{\nu_2}(X)$.

    Let $X\in \mathcal{X}(\mathcal{G})$ and $Y\in \mathcal{X}$ satisfy $X\succsim_1Y~(X\succ_1Y)$. We have 
    $$\gamma\circ\mathfrak{D}_{\nu_2}(X)=\gamma(X_{\#}\mathbb{P})=\gamma\circ\mathfrak{D}_{\nu_1}(X)\overset{(>)}{\geq} \gamma\circ\mathfrak{D}_{\nu_1}(Y)\geq \gamma\circ \mathfrak{D}_{\nu_2}(Y).$$
    Therefore, $X\succsim_2 Y~(X\succ_2Y)$.

    For the converse, assume that \eqref{eq:ep} holds true. Since \eqref{eq:ep} is stronger than the definition of comparative ambiguity aversion, $\succsim_2$ is more ambiguity averse than $\succsim_1$. Let $X,Y\in \mathcal{X}(\mathcal{G})$. If $X\succsim_1 Y$, then $X\succsim_2 Y$ by \eqref{eq:ep}.
    If $X\succsim_2 Y$, then $X\succsim_1 Y$ or else there would be a contradiction with \eqref{eq:ep}. Thus,
    $$X\succsim_1 Y\iff X\succsim_2 Y,~~X,Y\in \mathcal{X}(\mathcal{G}).$$
    Therefore, the $\succsim_1$-CES and the $\succsim_2$-CES both represent the Choquet ATE preference relation $\succsim_1$ on $\mathcal{X}(\mathcal{G})$. Therefore, by Proposition \ref{prop:riskfun}, the $\succsim_1$-CES equals the $\succsim_2$-CES.
\end{proof}

The following example shows that Property (RD) does not necessarily imply Property (D). In the following example, given $A\in \mathcal{F}$ and $\mu \in \Delta$, we use $\mu(A|\mathcal{H})$ to denote the conditional probability of $A$ under $\mu$ given the sub-$\sigma$-algebra $\mathcal{H}$.

\begin{example}
    Assume there exists a risk-conforming $\mu\in \Delta$ and a sub-$\sigma$-algebra $\mathcal{H}$, independent of $\mathcal{G}$ under $\mu$, such that there exists $A_0\in \mathcal{H}$ satisfying $\mu(A_0)= 1/2$. Let $\tilde{g}$ be a concave distortion function satisfying $\tilde{g}(1/2)>1/2$. Define the distortion function $\tilde{f}=\tilde{g}\circ \tilde{g}$. Let $g=\tilde{g}^{-1}$ and $f=\tilde{f}^{-1}$, it is clear that $g$ is convex. Also, $g(\tilde{f}(1/2))>1/2$. Define the capacity 
    $$\tilde{\nu}(A)=\int_{\Omega} f(\mu(A|\mathcal{H}))\d\mu,~~A\in \mathcal{F}.$$
    An application of both the monotone convergence theorem and the conditional monotone convergence theorem will verify that $\tilde{\nu}$ is continuous. As $\mathcal{G}$ and $\mathcal{H}$ are independent, it is straightforward to show that for all $A\in \mathcal{G}$, $\tilde{\nu}(A)=f(\mu(A))=f(\mathbb{P}(A))$. Also, for all $A\in \mathcal{H}$, $\tilde{\nu}(A)=\mu(A)$. Define the capacity $\nu$ by $\nu(A)=\tilde{f}\circ\tilde{\nu}(A)$ for all $A\in \mathcal{F}.$ It is clear that $\nu$ is risk conforming and $\nu(A)=\tilde{f}\circ\mu(A)$ for all $A\in \mathcal{H}$. As $f$ is continuous, $\nu$ is continuous. Define the Choquet ATE preference relation $\succsim$ by
    $$X\succsim Y\iff \gamma^{\mathrm{DU}}_g\circ\mathfrak{D}_{\nu}(X)\iff \gamma^{\mathrm{DU}}_g\circ\mathfrak{D}_{\nu}(Y).$$
    Since $g$ is convex, by \citet[Corollary 4.2]{MM04A}, $\succsim$ satisfies Property (RD). It is clear that $\id_{A_0}\simeq 
    \id_{A_0^c}$. We have
    $$
    \begin{aligned}
        \gamma^{\mathrm{DU}}_g&\circ\mathfrak{D}_{\nu}\left((1/2)\id_{A_0}+(1/2)\id_{A_0^c}\right)\\&=\gamma_g^{\mathrm{DU}}\left(\delta_{(1/2)}\right)=1/2<g\left(\tilde{f}(1/2)\right)=\gamma_g\left(\tilde{f}(\mu(A_0))\delta_1+\left(1-\tilde{f}(\mu(A_0))\right)\delta_0\right)=\gamma^{\mathrm{DU}}_g\circ\mathfrak{D}_{\nu}\left(\id_{A_0}\right).
    \end{aligned}$$
    Therefore, $\id_{A_0}\succ(1/2)\id_{A_0}+(1/2)\id_{A_0^c}$ and $\succsim$ does not satisfy Property (D).
\end{example}

\begin{proof}[Proof of Therorem \ref{theo:SAA}]
        Assume that $\mathfrak{D}_{\nu}$ is concave. Fix $A,B\in \mathcal{F}$. By Proposition \ref{prop:oper}, and the fact that $\id_{A\cup B}$ and $\id_{A\cap B}$ are comonotonic, we have 
        $$\mathfrak{D}_{\nu}((1/2)\id_A+(1/2)\id_B)=\mathfrak{D}_{\nu}((1/2)\id_{A\cap B}+(1/2)\id_{A\cup B})=(1/2)\otimes \mathfrak{D}_{\nu}(\id_{A\cup B})\oplus (1/2)\otimes \mathfrak{D}_{\nu}(\id_{A\cap B}).$$
        Furthermore, since $\mathfrak{D}_{\nu}$ is concave, it holds that 
        $$\mathfrak{D}_{\nu}((1/2)\id_A+(1/2)\id_B)\geq_{\mathrm{ssd}}(1/2)\otimes \mathfrak{D}_{\nu}(\id_{A})\oplus (1/2)\otimes \mathfrak{D}_{\nu}(\id_{B}).$$
        The following two properties of $\mathbb{E}$ are immediately clear. (a) $\mathbb{E}$ is monotonic with respect to $\geq_{\mathrm{ssd}}$. (b) For all $Q,P\in \mathcal{D}$ and non-negative $a,b\in \mathbb{R}$, it holds 
        $\mathbb{E}(a\otimes Q\oplus b\otimes P)=a\mathbb{E}(Q)+b\mathbb{E}(P).$ Let $A\in \mathcal{F}$, since $\mathfrak{D}_{\nu}(\id_A)=\nu(A)\delta_1+(1-\nu(A))\delta_0$, we have that $\mathbb{E}(\mathfrak{D}_{\nu}(\id_A))=\nu(A)$.
        Therefore, we have
        \begin{align*}
            (1/2)\nu(A\cap B)+(1/2)\nu(A\cup B)&=\mathbb{E}\left[(1/2)\otimes\mathfrak{D}_{\nu}(\id_{A\cup B})\oplus (1/2)\otimes \mathfrak{D}_{\nu}(\id_{A\cap B})\right]\\&=\mathbb{E}\left[\mathfrak{D}_{\nu}((1/2)\id_A+(1/2)\id_B)\right]
            \\&\geq\mathbb{E}\left[(1/2)\otimes \mathfrak{D}_{\nu}(\id_{A})\oplus (1/2)\otimes \mathfrak{D}_{\nu}(\id_{B})\right]\\&=(1/2)\nu(A)+(1/2)\nu(B).
        \end{align*}
        Thus, $\nu(A\cap B)+\nu(A\cup B)\geq \nu(A)+\nu(B)$ and $\nu$ is supermodular.

        For the converse, let $X,Y\in \mathcal{X}$ and $\lambda\in [0,1]$. Since $\nu$ is supermodular, $\mathrm{H}_\alpha^\nu$ is concave for all $\alpha \in (0,1)$; see Appendix \ref{app:DC}. Therefore, by Proposition \ref{prop:quant}, we have
        \begin{align*}
            \int_0^\alpha \mathrm{VaR}_\beta\left(\mathfrak{D}_\nu\left(\lambda X+(1-\lambda)Y\right)\right)\d \beta&=\int_0^\alpha \mathrm{G}_\beta^\nu\left(\lambda X+(1-\lambda)Y\right)\d \beta\\&=\alpha \mathrm{H}_\alpha^\nu\left(\lambda X+(1-\lambda)Y\right)\\&\geq \alpha\lambda\mathrm{H}_\alpha^\nu(X)+\alpha(1-\lambda)\mathrm{H}_\alpha^\nu(Y)\\&=  \lambda\int_0^\alpha \mathrm{G}_\beta^\nu\left(X\right)\d \beta+(1-\lambda)\int_0^\alpha \mathrm{G}_\beta^\nu\left(Y\right)\d \beta\\&=\int_0^\alpha \lambda\mathrm{VaR}_\beta\left(\mathfrak{D}_\nu(X)\right)+(1-\lambda)\mathrm{VaR}_\beta\left(\mathfrak{D}_\nu(Y)\right)\d \beta\\&=\int_0^\alpha \mathrm{VaR}_\beta\left(\lambda\otimes \mathfrak{D}_{\nu}(X)\oplus(1-\lambda)\otimes \mathfrak{D}_{\nu}(Y)\right)\d \beta.
        \end{align*}
        Thus, by \citet[Theorem 4.A.3]{SS07A}, $\mathfrak{D}_{\nu}$ is concave.
    \end{proof}

\begin{proof}[Proof of the claims in Example \ref{ex:counter2}]
    Since $\mu_1$ and $\mu_2$ are supermodular, $\mathrm{H}_\alpha^{\mu_1}$ and $\mathrm{H}_\alpha^{\mu_2}$ are concave for all $\alpha \in (0,1)$; see Appendix \ref{app:DC}. Let $X,Y\in \mathcal{X}$ and $\lambda\in [0,1]$.  we have
        \begin{align*}
            \int_0^\alpha \mathrm{VaR}_\beta&\left(\mathfrak{D}\left(\lambda X+(1-\lambda)Y\right)\right)\d \beta\\&=(1/2)\int_0^\alpha \mathrm{VaR}_\beta\left(\mathfrak{D}_{\mu_1}\left(\lambda X+(1-\lambda)Y\right)\right)\d \beta+(1/2)\int_0^\alpha \mathrm{VaR}_\beta\left(\mathfrak{D}_{\mu_2}\left(\lambda X+(1-\lambda)Y\right)\right)\d \beta      \\&=(1/2)\int_0^\alpha \mathrm{G}_\beta^{\mu_1}\left(\lambda X+(1-\lambda)Y\right)\d \beta+(1/2)\int_0^\alpha \mathrm{G}_\beta^{\mu_2}\left(\lambda X+(1-\lambda)Y\right)\d \beta\\&=(1/2)\alpha \mathrm{H}_\alpha^{\mu_1}\left(\lambda X+(1-\lambda)Y\right)+(1/2)\alpha \mathrm{H}_\alpha^{\mu_2}\left(\lambda X+(1-\lambda)Y\right)\\&\geq (1/2)\alpha\lambda\mathrm{H}_\alpha^{\mu_1}(X)+(1/2)\alpha(1-\lambda)\mathrm{H}_\alpha^{\mu_1}(Y)+(1/2)\alpha\lambda\mathrm{H}_\alpha^{\mu_2}(X)+(1/2)\alpha(1-\lambda)\mathrm{H}_\alpha^{\mu_2}(Y)\\&= (1/2)\lambda\int_0^\alpha \mathrm{G}_\beta^{\mu_1}\left(X\right)\d \beta+(1/2)(1-\lambda)\int_0^\alpha \mathrm{G}_\beta^{\mu_1}\left(Y\right)\d \beta\\&~~~+(1/2)\lambda\int_0^\alpha \mathrm{G}_\beta^{\mu_2}\left(X\right)\d \beta+(1/2)(1-\lambda)\int_0^\alpha \mathrm{G}_\beta^{\mu_2}\left(Y\right)\d \beta\\&=(1/2)\lambda\int_0^\alpha \mathrm{VaR}_\beta\left(\mathfrak{D}_{\mu_1}(X)\right)+\mathrm{VaR}_\beta\left(\mathfrak{D}_{\mu_2}(X)\right)\d \beta\\&~~~+(1/2)(1-\lambda)\int_0^\alpha \mathrm{VaR}_\beta\left(\mathfrak{D}_{\mu_1}(Y)\right)+\mathrm{VaR}_\beta\left(\mathfrak{D}_{\mu_2}(Y)\right)\d \beta\\&=\int_0^\alpha \lambda\mathrm{VaR}_\beta\left(\mathfrak{D}(X)\right)+(1-\lambda)\mathrm{VaR}_\beta\left(\mathfrak{D}(Y)\right)\d \beta\\&=\int_0^\alpha \mathrm{VaR}_\beta\left(\lambda\otimes \mathfrak{D}(X)\oplus(1-\lambda)\otimes \mathfrak{D}(Y)\right)\d \beta.
        \end{align*}
    Thus, by \citet[Theorem 4.A.3]{SS07A}, $\mathfrak{D}$ is concave.

      For the sake of contradiction, assume that $\mathfrak{D}$ is ambiguity averse, that is, assume there exists a $\mu\in \Delta$ such that $\mathfrak{D}_\mu(X)\geq_{\mathrm{fsd}}\mathfrak{D}(X)$ for all $X\in \mathcal{X}.$
    Find $A\in \mathcal{F}$ satisfying $\mu_1(A)>\mu_2(A).$ We claim that $\mu(A)\geq \mu_1(A)$. If $\mu(A)=1$, this is clear. If $\mu(A)< 1$, define $\alpha'=1-\mu(A)$. Given $\epsilon\in (0,\alpha')$, it is straightforward to show that $\mathrm{VaR}_{\alpha'-\epsilon}(\mathfrak{D}_{\mu}(\id_A))=0$. Therefore, $0=\mathrm{VaR}_{\alpha'-\epsilon}(\mathfrak{D}(\id_A))$. Since for all $\alpha$,
    $$\mathrm{VaR}_{\alpha}(\mathfrak{D}(\id_A))=\begin{cases}
        1&\text{if}~\alpha\in [1-\mu_2(A),1)\\
        1/2&\text{if}~\alpha\in [1-\mu_1(A),1-\mu_2(A))\\
        0&\text{if}~\alpha\in (0,1-\mu_1(A)),
    \end{cases}$$
    we have $ 1-\mu_1(A)>\alpha'-\epsilon$. Since $\epsilon$ was arbitrary, $ 1-\mu_1(A)\geq \alpha'$. Therefore, $\mu(A)\geq \mu_1(A)$. A similar argument will show that $\mu(A^c)\geq \mu_2(A^c)$. Therefore,
    $$1=\mu(A)+\mu(A^c)\geq \mu_1(A)+\mu_2(A^c)>\mu_2(A)+\mu_2(A^c)=1,$$
    a contradiction. Thus, by contradiction, $\mathfrak{D}$ is not ambiguity averse.
    \end{proof} 

\begin{lemma}\label{lem:Equiv}
    Let $\nu$ be a continuous capacity. Then $\nu$ is exact if and only if $\mathfrak{D}_{\nu}=\mathfrak{D}_{\mathfrak{C}(\nu)}$.
\end{lemma}

\begin{proof}
    Assume that $\nu$ is exact. Given $X\in \mathcal{X}$, by \citet[Proposition 1]{MWW25A}, we have
    $$S_{\mathfrak{D}_{\mathfrak{C}(\nu)}(X)}(x)=\min_{\mu\in \mathfrak{C}(\nu)}\mu(X>x)=\nu(X>x)=S_{\mathfrak{D}_\nu(X)}(x),~~x\in \mathbb{R}.$$
    Therefore, $\mathfrak{D}_{\nu}=\mathfrak{D}_{\mathfrak{C}(\nu)}$. 
    
    For the converse, let $A\in \mathcal{F}$. By \citet[Proposition 1]{MWW25A}, we have $$\nu(A)=\nu(\id_A>0)=S_{\mathfrak{D}_\nu(\id_A)}(0)=S_{\mathfrak{D}_{\mathfrak{C}(\nu)}(\id_A)}(0)=\min_{\mu\in \mathfrak{C}(\nu)}\mu(\id_A>0)=\min_{\mu\in \mathfrak{C}(\nu)}\mu(A).$$
    Thus, $\nu$ is exact.
\end{proof}

\begin{proof}[Proof of Proposition \ref{prop:weakerDiv}]
   Assume that $\mathfrak{D}_{\nu}=\mathfrak{D}_{\mathfrak{C}(\nu)}$. Let $X_1,\dots,X_n\in\mathcal{X}$ and $\lambda_1,\dots,\lambda_n\in [0,1]$ with $\sum_{k=1}^n \lambda_k=1$ 
    satisfy $\sum_{k=1}^n \lambda_k X_k=xAy$ for some $A\in \mathcal{F}$ and $x,y\in \mathbb{R}$. By Lemma \ref{lem:Equiv}, $\nu$ is exact. Therefore, there exists $\mu'\in \mathfrak{C}(\nu)$ satisfying $\mu'(A)=\nu(A).$ Without loss of generality, we may assume $x\geq y$. It is clear that $\mathfrak{D}_\nu(xAy)=\mathfrak{D}_{\mu'}(xAy).$ Also, by the comonotonic sum inequality, we have
    $$\mathfrak{D}_{\mu'}(xAy)=\mathfrak{D}_{\mu'}\left(\sum_{k=1}^n \lambda_k X_k\right)\geq_{\mathrm{ssd}} \bigoplus_{k=1}^n \lambda_k\otimes \mathfrak{D}_{\mu'}(X_k).$$
    Finally, since $\mu'\in \mathfrak{C}(\nu)$, we have that $\mathfrak{D}_{\mu'}(X_k)\geq_{\mathrm{fsd}}\mathfrak{D}_\nu(X_k)$ for all $k$. Therefore,
    $$\bigoplus_{k=1}^n \lambda_k\otimes \mathfrak{D}_{\mu'}(X_k)\geq_{\mathrm{fsd}}\bigoplus_{k=1}^n \lambda_k\otimes \mathfrak{D}_{\nu}(X_k).$$
    As $\geq_{\mathrm{fsd}}$ is stronger than $\geq_{\mathrm{ssd}}$, we have that \eqref{eq:newDiv} is true.

    For the converse, by Lemma \ref{lem:Equiv}, we only need to show that $\nu$ is exact. Fix $A_1,\dots,A_n,A\in \mathcal{F}$ and  non-negative $c_1,\dots,c_n,c\in\mathbb{R}$ satisfying
    $\id_A\geq\sum_{k=1}^n c_k\id_{A_k}-c.$ Find $A_{n+1},\dots A_{m}\in \mathcal{F}$ and non-negative $c_{n+1},\dots,c_m\in\mathbb{R}$ such that $\id_A=\sum_{k=1}^m c_k\id_{A_k}-c,$ where $m\geq n$. Define $a=\sum_{k=1}^m c_k$. Therefore, by \eqref{eq:newDiv},
    $$\begin{aligned}
        (1/a)\otimes\left[(\nu(A))\delta_{c+1}+(1-\nu(A))\delta_{c}\right]&=(1/a)\otimes\mathfrak{D}_{\nu}\left((c+1)Ac\right)\\&\geq_{\mathrm{ssd}}\bigoplus_{k=1}^m (c_k/a)\otimes \mathfrak{D}_\nu(\id_{A_k})\\&=\bigoplus_{k=1}^m (c_k/a)\otimes \left[(\nu(A_k))\delta_{1}+(1-\nu(A_k))\delta_{0}\right].
    \end{aligned}$$
    Since $\mathbb{E}\left( (1/a)\otimes\left[(\nu(A))\delta_{c+1}+(1-\nu(A))\delta_{c}\right]\right)=(1/a)(\nu(A)(c+1)+(1-\nu(A))c)=(1/a)(\nu(A)+c)$, $\mathbb{E}\left(\bigoplus_{k=1}^m (c_k/a)\otimes \left[(\nu(A_k))\delta_{1}+(1-\nu(A_k))\delta_{0}\right]\right)=\sum_{k=1}^m (c_k/a)\nu(A_k)$, and $\mathbb{E}$ is monotone with respect to $\geq_{\mathrm{ssd}}$, we have $\nu(A)\geq \sum_{k=1}^m c_k\nu(A_k)-c\geq \sum_{k=1}^n c_k\nu(A_k)-c$. Therefore, by \cite{S72}, $\nu$ is exact. 
\end{proof}

\begin{proof}[Proof of Theorem \ref{theo:robust}]
    Since $\succsim$ is a Choquet ATE preference relation, we have that
    \begin{equation}\label{eq:need1}
        X\succsim Y\iff \gamma\circ\mathfrak{D}_{\nu}(X)\geq \gamma\circ\mathfrak{D}_{\nu}(Y),~~X,Y\in \mathcal{X},
    \end{equation}
    where $\gamma$ is the $\succsim$-CES.

    Assume that $\succsim$ is a distributionally robust preference relation. Therefore, by definition, there exists a collection of risk-conforming probability measures $\mathcal{Q}$ and a statistic $\tilde{\gamma}$ such that 
    \begin{equation}
        \label{eq:need2}
        X\succsim Y\iff \min_{\mu\in \mathcal{Q}}\tilde{\gamma}(X_{\#}\mu)\geq\min_{\mu\in \mathcal{Q}}\tilde{\gamma}(Y_{\#}\mu),~~X,Y\in \mathcal{X}.
    \end{equation}
    Without loss of generality, we may assume that $\tilde{\gamma}$ is a certainty-equivalent statistic. Fix $Q\in \mathcal{D}$. Since $(\Omega,\mathcal{G},\mathbb{P})$ is atomless, we can find $X\in \mathcal{X}(\mathcal{G})$ such that $Q=X_{\#}\mathbb{P}=\mathfrak{D}_\nu(X)$. By \eqref{eq:need1}, we have that $X\simeq\gamma(Q)$. Thus, by \eqref{eq:need2}, we have $\gamma(Q)=\tilde{\gamma}(Q)$. Since $Q$ was arbitrary, we have that $\gamma=\tilde{\gamma}$. Given $X\in \mathcal{X}$, by \eqref{eq:need1}, $X\simeq \gamma\circ\mathfrak{D}_{\nu}(X)$. Thus, by \eqref{eq:need2}, we have \begin{equation}
        \label{eq:need3}
        \min_{\mu\in \mathcal{Q}}\gamma(X_{\#}\mu)=\gamma\circ\mathfrak{D}_{\nu}(X).
    \end{equation} Define $\varphi:[0,1]\to [0,1]$ by $\varphi(x)=\gamma(x\delta_1+(1-x)\delta_0)$. Since $\gamma$ is a certainty-equivalent statistic, it is easy to show that $\varphi$ is a distortion function. Let $A\in \mathcal{F}$, by \eqref{eq:need3} we have
    $$\begin{aligned}
        \varphi(\nu(A))=\gamma(\nu(A)\delta_1+(1-\nu(A))\delta_0)&=\gamma\circ\mathfrak{D}_\nu(\id_A)\\&=\min_{\mu\in \mathcal{Q}}\gamma((\id_A)_{\#}\mu)\\&=\min_{\mu\in \mathcal{Q}}\gamma(\mu(A)\delta_1+(1-\mu(A))\delta_0)\\&=\min_{\mu\in \mathcal{Q}}\varphi(\mu(A))=\varphi\left(\min_{\mu\in \mathcal{Q}}\mu(A)\right).
    \end{aligned}
    $$
    Since $\varphi$ is invertible, we have $\nu(A)=\min_{\mu\in \mathcal{Q}}\mu(A)$. Additionally, we have $\mu(A)\geq \nu(A)$ for all $\mu\in \mathcal{Q}$ and $A\in \mathcal{F}$. Thus $\mathcal{Q}\subseteq \mathfrak{C}(\nu).$ We claim that for each $X\in \mathcal{X}$, there exists $\mu_0\in \mathcal{Q}$ such that $\mathfrak{D}_\nu(X)=X_{\#}\mu_0.$ To see this note that for all $\mu\in \mathcal{Q}$, we have $\mu(X>x)\geq \nu(X>x)$ for all $x\in \mathbb{R}$. Therefore, $$X_{\#}\mu \geq_{\mathrm{fsd}}\mathfrak{D}_{\nu}(X),~~\mu\in \mathcal{Q}.$$
    Additionally, by \eqref{eq:need3}, there exists $\mu_0\in \mathcal{Q}$ such that $\gamma(X_{\#}\mu_0)=\gamma\circ\mathfrak{D}_{\nu}(X).$ If $X_{\#}\mu_0 >_{\mathrm{fsd}}\mathfrak{D}_{\nu}(X)$, it would hold that $\gamma(X_{\#}\mu_0)>\gamma\circ\mathfrak{D}_{\nu}(X),$ which is a contradiction. Therefore, $\mathfrak{D}_\nu(X)=X_{\#}\mu_0.$ Given $A,B\in \mathcal{F}$ satisfying $A\subseteq B$, we have that there exists a $\mu_0\in \mathcal{Q}$ such that
    $$\begin{aligned}
        \mu(A)\delta_2&+(\mu(B)-\mu(A))\delta_1+(1-\mu(B))\delta_0\\&=(\id_A+\id_B)_{\#}\mu_0=\mathfrak{D}_{\nu}(\id_A+\id_B)=\nu(A)\delta_2+(\nu(B)-\nu(A))\delta_1+(1-\nu(B))\delta_0.
    \end{aligned}$$
    Therefore, $\nu(A)=\mu(A)$ and $\nu(B)=\mu(B)$. Since $\mu\in \mathfrak{C}(\nu)$, by \citet[Theorem 4.7]{MM04A}, we have that $\nu$ is supermodular.

    Conversely, assume that $\nu$ is supermodular. We claim that \begin{equation}
        \label{eq:need4}
        \gamma\circ\mathfrak{D}_{\nu}(X)=\min_{\mu\in \mathfrak{C}(\nu)}\gamma(X_{\#}\mu),~~X\in \mathcal{X}.
    \end{equation}
    Since for all $\mu\in \mathfrak{C}(\nu)$, $\mu$ is risk conforming, \eqref{eq:need4}
     shows that $\succsim$ is a distributionally robust preference relation. To show \eqref{eq:need4}, start by fixing $X\in \mathcal{X}$. As $\{X>x\}_{x\in \mathbb{R}}$ is a chain, by \citet[Theorem 4.7]{MM04A}, there exists $\mu_0\in \mathfrak{C}(\nu)$ such that for all $\mu\in \mathfrak{C}(\nu)$, we have
     $$\mu(X>x)\geq\nu(X>x)=\mu_0(X>x),~~x\in \mathbb{R}.$$ Therefore, for all $\mu\in \mathfrak{C}(\nu)$, $X_{\#}\mu\geq_{\mathrm{fsd}}\mathfrak{D}_\nu(X)=X_{\#}\mu_0$. Thus, 
     $$\gamma(X_{\#}\mu)\geq \gamma\circ\mathfrak{D}_\nu(X)=\gamma(X_{\#}\mu_0),~~\mu\in \mathfrak{C}(\nu).$$
     Finally, as $X$ was arbitrary, \eqref{eq:need4} holds.
\end{proof}

\end{document}